\documentclass[12pt,a4paper]{amsart}
\usepackage[english,french]{babel}
\usepackage[latin1]{inputenc}
\usepackage{color}
\usepackage{amsfonts}
\usepackage{amssymb}
\usepackage{newlfont}
\usepackage{mathrsfs}
\usepackage{float}
\usepackage{enumerate}
\newcommand\blfootnote[1]{%
  \begingroup
  \renewcommand\thefootnote{}\footnote{#1}%
  \addtocounter{footnote}{-1}%
  \endgroup
}

\definecolor{liens}{rgb}{1,0,0}
\usepackage[colorlinks=true, linkcolor=blue, 
hyperfootnotes=true,citecolor=blue,urlcolor=black]{hyperref}

\usepackage{tikz}
\usetikzlibrary{shapes.geometric}

\newtheorem{theo}{Theorem}

\newtheorem{prop}[theo]{Proposition}

\newtheorem{example}[theo]{Example}
\newtheorem{lem}[theo]{Lemma}

\newtheorem{remark}[theo]{Remark}

\newcommand{\Z}{\mathbb Z}
\newcommand{\Q}{\mathbb{Q}}

\newcommand{\R}{\mathbb R}
\newcommand{\C}{\mathbb C}

\newcommand{\ev}[1]{\operatorname{ev}_{\lambda=#1}}

\newcommand{\val}{\operatorname{val}_{z}}

\newcommand{\Rla}{\mathscr{R}_{\lambda}}
\newcommand{\Rspe}{\mathscr{R}}

\newcommand{\New}{\mathcal{N}}

\newcommand{\ellmahl}{p}

\newcommand{\malop}[1]{\phi_{#1}}

\newcommand{\logm}{\ell}

\newcommand{\cld}{\operatorname{cld_{z}}}

\newcommand{\dqhahn}{\mathcal{D}_{\Hahn}}

\newcommand{\univ}[1]{\mathscr{U}}

\newcommand{\Hahn}{\mathscr{H}}
\newcommand{\Hahnlamb}{\Hahn_{\C(\lambda)}}
\newcommand{\supp}{\operatorname{supp}}

\usepackage[all]{xy}

\title[Frobenius method]{Frobenius method for Mahler equations}
\author{Julien Roques}
\address{Universite Claude Bernard Lyon 1, CNRS, Ecole Centrale de Lyon, INSA Lyon, Universit\'e Jean Monnet, ICJ UMR5208, 69622 Villeurbanne, France.}
\email{Julien.Roques@univ-lyon1.fr}
\thanks{}

\date{\today}

\begin{document}
\selectlanguage{english}
\sloppy

\blfootnote{This is the accepted version of the following article : {\it Julien Roques. Frobenius method for Mahler equations. J. Math. Soc. Japan 76 (1) 229 - 268, January, 2024},  which has been published in final form at \url{https://doi.org/10.2969/jmsj/89258925}.}

\begin{abstract}
Using Hahn series, one can attach to any linear Mahler equation a basis of solutions at $0$ reminiscent of the solutions of linear differential equations at a regular singularity. We show that such a basis of solutions can be produced by using a variant of Frobenius method. 
\end{abstract}



\maketitle
%\tableofcontents

\section{Introduction}

A linear Mahler equation with coefficients in $\C(z)$ is a functional equation of the form
\begin{equation}\label{eq mahl intro}
 a_{n}(z) y(z^{\ellmahl^{n}}) + a_{n-1}(z) y(z^{\ellmahl^{n-1}}) + \cdots + a_{0}(z) y(z) = 0
\end{equation}
for some $\ellmahl \in \Z_{\geq 2}$, some $n \in \Z_{\geq 0}$ and some $a_{0}(z),\ldots,a_{n}(z) \in \C(z)$ such that $a_{0}(z)a_{n}(z) \neq 0$.  These equations appear for instance in the theory of automatic sequences~: the generating series of any automatic sequence satisfies a nontrivial linear Mahler equation with coefficients in $\C(z)$. The following list of references gives an idea of the many facets of the theory of Mahler equations \cite{MahlerK1929Arith,MahlerUber1930,Kubota,LoxtVanderPort,MasserVanishTheoPowSeries,RandeThese,DumasThese,BeckerKReg,NishiokaLNM1631,DF96,ZannierOFE,CorvajaZannierSNA,AS03,pellarinAIMM,NG,NGT,PhilipponGaloisNA,SchafkeSinger,AlgIndMahlFunc,BorisAboutMahler,BorisFaverjonMethodeMahler17,BorisFaverjonMethodeMahler,DHRMahler,CompSolMahlEq,BeckerConj,FernandesMMCarNN,RoquesLSMS,PouletDensity}. 

In \cite{RoquesLSMS}, we have shown the relevance of Hahn series for the study of Mahler equations. It follows from \cite[Theorem 2]{RoquesLSMS} that \eqref{eq mahl intro} has $n$ $\C$-linearly independent solutions $y_{1}(z),\ldots,y_{n}(z)$ of the form 
\begin{equation}\label{form yi intro}
 y_{i}(z) = \sum_{(c,j) \in \C^{\times} \times \Z_{\geq 0}}  f_{i,c,j}(z) e_{c}(z) \logm(z)^{j} 
\end{equation}
where the terms involved in this finite sum have the following properties~:
\begin{itemize}
\item the $f_{i,c,j}(z)$ belong to the field $\Hahn$ of Hahn series with coefficients in $\C$ and value group $\Q$; 
 \item $e_{c}(z)$ satisfies $e_{c}(z^{\ellmahl})=ce_{c}(z)$;
\item $\logm(z)$ satisfies $\logm(z^{\ellmahl})=\logm(z)+1$. 
\end{itemize}

\begin{remark}
If $c\in\C^{\times}\setminus\{1\}$, then the equation $y(z^{\ellmahl})=cy(z)$ has no nonzero solution in $\Hahn$. Similarly, the equation $y(z^{\ellmahl})=y(z)+1$ has no solution in $\Hahn$. 
In this paper, the $e_{c}(z)$ and $\ell(z)$ will be ``symbols'', {\it i.e.}, they will be constructed algebraically as elements of a certain difference field extension of $\Hahn$; see Section~\ref{sec hahn}. However, note that one can find ``true functions'' solutions of these equations; for instance, $(\log z)^{\frac{\log c}{\log \ellmahl}}$ is a solution of $y(z^{\ellmahl})=cy(z)$ and $\frac{\log \log z}{\log \ellmahl}$ is a solution of $y(z^{\ellmahl})=y(z)+1$. 
\end{remark}

The shape of the $y_{i}(z)$ is reminiscent of the shape of the solutions of a linear differential equation at a {\it regular singularity}. Indeed, we recall that  
any linear differential equation of order $n$ with coefficients in $\C(z)$ having at worst a regular singularity at $0$ has $n$ $\C$-linearly independent solutions $\widetilde{y}_{1}(z),\ldots,\widetilde{y}_{n}(z)$ of the form 
\begin{equation}\label{form yi diffel intro}
 \widetilde{y}_{i}(z) = \sum_{(\alpha,j) \in \C \times \Z_{\geq 0}}  \widetilde{f}_{i,\alpha,j}(z) z^{\alpha} \log(z)^{j} 
\end{equation}
where the $\widetilde{f}_{i,\alpha,j}(z)$ involved in this finite sum belong to the field of formal Laurent series $\C((z))$.  
It is thus natural to wonder wether a basis of solutions of the form \eqref{form yi intro} can be derived by using a variant of the celebrated Frobenius method \cite{FrobMethodFrob}. 
The aim of the present paper is to bring a positive answer to this question. 

\begin{remark}
The original Frobenius method produces solutions of linear differential equations at a regular singular point. We remind that the fact that a given point is regular singular can be read from a certain Newton polygon: it corresponds to the case where this polygon has only one (finite) slope and that this slope is equal to $0$. See Section~\ref{classical frob method} for references. 
It turns out that such restrictions are not necessary in the Mahler case: the method introduced in the present paper works for any Mahler equation with coefficients in the field $\Hahn$ of Hahn series. 
\end{remark}

\subsection{Organization of the paper}
Section \ref{classical frob method} provides an overview of the classical Frobenius method for linear differential equations. An extension of Frobenius method to Mahler equations is presented in Section~\ref{sec: statement frob}. Its statement requires notions and notations introduced in 
Section \ref{sec hahn} and Section \ref{sec: newt pol} (Section \ref{sec hahn} introduces rings where the calculations of the method will take place, Section \ref{sec: newt pol} introduces the notions of Newton polygons, slopes and exponents for Mahler equations). A simple but non trivial example of application of the method is given in Section \ref{sec: main example}. The proofs of the unjustified statements of Section~\ref{sec: statement frob} occupy Sections \ref{sec: fact} and \ref{sec:frob method}. 

\vskip 10 pt 
\noindent {\bf Acknowledgement.} I would like to thank the anonymous referee for its useful comments. This work was supported by the ANR De rerum natura project, grant ANR-19-CE40-0018 of the French Agence Nationale de la Recherche.

\section{Frobenius method for regular singular linear differential equations}\label{classical frob method} In the theory of differential equations, Frobenius left his name to an important method for computing a basis of solutions of any linear differential equation at a regular singularity. Here is an outline of this method; we refer to  \cite{CoddingtonLevinsonTODE,fromgausstopainleve,FrobMethEncMath,FrobMethodFrob} for more details. 

Consider a (formal) linear differential operator 
\begin{equation}\label{op diffal intro}
\mathcal{L} = \delta^{n} + b_{n-1}(z) \delta^{n-1} + \cdots +b_{0}(z)
\end{equation}
with coefficients $b_{0}(z),\ldots,b_{n-1}(z)$ in the field $\C((z))$ of formal Laurent series written in terms of Euler operator $\delta=z\frac{\mathrm{d} \ }{\mathrm{d} z}$.
As mentioned above, we will focus our attention on the regular singular case, {\it i.e.}, from now on, we assume that $b_{0}(z),\ldots,b_{n-1}(z)$ belong to the ring $\C[[z]]$ of formal power series (see \cite[Definition~3.14]{VdPS}). 

\begin{remark}\label{rem: reg sing and newton diffal eq}
One classically attaches to $\mathcal{L}$ a Newton polygon $\mathcal{N}(\mathcal{L}) \subset \R^{2}$ and a (finite) set of slopes $\mathcal{S}(\mathcal{L}) \subset \Q$; see \cite{PerronULD,HsiehSibuya} and \cite[$\S$3.3]{VdPS}. The following properties are equivalent: 
\begin{itemize}
 \item $\mathcal{L}$ is regular singular;
 \item $b_{0}(z),\ldots,b_{n-1}(z)$ belong to $\C[[z]]$;
 \item  $\mathcal{S}(\mathcal{L})=\{0\}$.
\end{itemize}
 See \cite[Definition~3.14]{VdPS} and \cite[Exercise~3.47]{VdPS}.  
\end{remark}

The indicial polynomial of $\mathcal{L}$ at $0$ is the polynomial of degree $n$ given by   
 $$
 \chi(\mathcal{L};X) = X^{n} + b_{n-1}(0)X^{n-1}+\cdots+b_{0}(0).
 $$ 
The complex roots of $\chi(\mathcal{L};X)$ are called the exponents of $\mathcal{L}$ at $0$. The multiplicity of such an exponent $\alpha$ is by definition its multiplicity as a root of $\chi(\mathcal{L};X)$ and will be denoted by $m_{\alpha}$. If $\alpha \in \C$ is not an exponent of $\mathcal{L}$ at $0$ then we set $m_{\alpha}=0$. 

The Frobenius method attaches $m_{\alpha}$ solutions of $\mathcal{L}$ to any exponent $\alpha$  as follows.  
The fundamental idea is to introduce a parameter $\lambda$.  One can prove that:
\begin{itemize}
 \item there exists a unique $g_{\alpha}(\lambda,z) \in 1+z\C(\lambda)[[z]]$ such that 
$$
\mathcal{L}(g_{\alpha}(\lambda,z)z^{\lambda}) = (\lambda-\alpha)^{s_{\alpha}+m_{\alpha}}z^{\lambda} 
$$
where 
$$
s_{\alpha} = \sum_{\beta \in \alpha + \Z_{\geq 1}} m_{\beta};
$$
\item the coefficients of $g_{\alpha}(z,\lambda)$ have no pole at $\lambda=\alpha$;
\item the derivatives 
\begin{eqnarray*}
 &\partial_{\lambda}^{s_{\alpha}} (g_{\alpha}(\lambda,z)z^{\lambda})_{\vert \lambda=\alpha},& \\ 
 &\partial_{\lambda}^{s_{\alpha}+1} (g_{\alpha}(\lambda,z)z^{\lambda})_{\vert \lambda=\alpha},& \\
  &\ldots& \\
& \partial_{\lambda}^{s_{\alpha} + m_{\alpha}-1} (g_{\alpha}(\lambda,z)z^{\lambda})_{\vert \lambda=\alpha}&
\end{eqnarray*}
are $m_{\alpha}$ $\C$-linearly independent solutions of $\mathcal{L}$, where we have used the notation $\partial_{\lambda}=\frac{\partial \ }{\partial \lambda}$. 
\end{itemize}
 We then have~: 

\begin{theo}[{Frobenius \cite{FrobMethodFrob}}]\label{theo:theo frob intro}
 We have attached to any exponent $\alpha$  and to any $m \in \{0,\ldots,m_{\alpha}-1\}$ a solution 
$$
y_{\alpha,m}=\partial_{\lambda}^{s_{\alpha}+m} (g_{\alpha}(\lambda,z)z^{\lambda})_{\vert \lambda=\alpha}
$$ 
of $\mathcal{L}$.   
We obtain in this way a family of $n$ $\C$-linearly independent solutions of $\mathcal{L}$. 
\end{theo}

Using the Leibniz rule, we see that 
\begin{eqnarray*}
 y_{\alpha,m} & \in & \operatorname{Span}_{\C((z))} (\partial_{\lambda}^{0} (z^{\lambda})_{\vert \lambda=\alpha},\partial_{\lambda}^{1} (z^{\lambda})_{\vert \lambda=\alpha},\ldots,\partial_{\lambda}^{s_{\alpha}+m} (z^{\lambda})_{\vert \lambda=\alpha}) \\ 
 &=& \operatorname{Span}_{\C((z))} (z^{\alpha},\log(z)z^{\alpha},\ldots,(\log(z))^{s_{\alpha}+m}z^{\alpha}). 
\end{eqnarray*}

\begin{remark}
The solutions attached to $\alpha$ involve $\log(z)$ if $\alpha$ is an exponent of $\mathcal{L}$ of multiplicity $\geq 2$ ({\it i.e.}, $m_{\alpha} \geq 2$). But, this is not the only case : logarithms may also appear in the solutions attached to $\alpha$ if $\alpha + \Z_{\geq 1}$ contains at least one exponent of $\mathcal{L}$ ({\it i.e.}, $s_{\alpha} \geq 1$). 
\end{remark}

\section{Hahn series and other useful rings} \label{sec hahn}

The calculations involved in the classical Frobenius method outlined in Section \ref{classical frob method} take place in the differential ring 
$$
\C(\lambda)((z))[z^{\lambda},\partial_{\lambda} (z^{\lambda}),\partial_{\lambda}^{2}(z^{\lambda}),\ldots]
$$ 
equipped with the derivatives $\partial_{z}=\frac{\partial \ }{\partial z}$ and $\partial_{\lambda}=\frac{\partial \ }{\partial \lambda}$. 
This section aims at introducing an avatar of this differential ring in which the calculations of our variant of  Frobenius method for Mahler equations will take place. 

\subsection{The ring of Hahn series} The first fundamental fact is that the classical power series will not be sufficient for our purpose; we will need Hahn series \cite{HahnSeriesHahn}.

Let $R$ be a ring. We denote by 
$
\Hahn_{R}%=K ((z^{\Q}))
$ 
the ring of Hahn series with coefficients in $R$ and with value group $\Q$. An element of $\Hahn_{R}$ is a sequence $(f_{\gamma})_{\gamma \in \Q} \in R^{\Q}$ whose support 
$$
\supp((f_{\gamma})_{\gamma \in \Q})=\{\gamma\in \Q \ \vert \ f_{\gamma}\neq 0\}
$$ 
is well-ordered, {\it i.e.}, any nonempty subset of $\supp(f)$ has a least element. An element $(f_{\gamma})_{\gamma \in \Q}$ of $
\Hahn_{R}
$
is usually (and will be) denoted by 
 $$
 f(z)=\sum _{{\gamma \in \Q }}f_{\gamma}z^{\gamma}. 
 $$
The sum and product of two elements 
$f(z)=\sum _{{\gamma\in \Q }}f_{\gamma}z^{\gamma}$ and 
$g(z)=\sum _{{\gamma\in \Q }}g_{\gamma}z^{\gamma}$ of $\Hahn_{R}$ 
are respectively defined by
$$
f(z)+g(z)=\sum _{{\gamma\in \Q }}(f_{\gamma}+g_{\gamma})z^{\gamma}
$$
and
$$
f(z)g(z)=\sum _{{\gamma\in \Q }}\left(\sum _{{\gamma'+\gamma''=\gamma}}f_{{\gamma'}}g_{{\gamma''}}\right)z^{\gamma}.
$$
(Note that there are only finitely many $(\gamma',\gamma'') \in \Q \times \Q$ such that $\gamma'+\gamma''=\gamma$ and $f_{{\gamma'}}g_{{\gamma''}}\neq 0$.)

If $R$ is an integral domain (resp. a field), then $\Hahn_{R}$ is an integral domain (resp. a field) as well.  

For $R=\C$, we will use the shorthand notation 
$$
\Hahn=\Hahn_{\C}. 
$$

The $z$-adic valuation on $\Hahn_{R}$ will be denoted by
$$
\val : \Hahn_{R} \rightarrow \Q \cup \{+\infty\}.
$$
It is given, for any $f(z) \in \Hahn_{R}$, by 
$$
\val(f(z))=\min \supp(f(z))
$$
with the convention $\min \emptyset = +\infty$. 

For $f(z)=\sum _{{\gamma\in \Q }}f_{\gamma}z^{\gamma} \in \Hahn_{R} \setminus \{0\}$, we will denote by $\cld(f(z))$ the coefficient of $z^{\val(f(z))}$ in $f(z)$~: 
$$\cld(f(z)) = f_{\val(f(z))} \in R \setminus \{0\}.$$ 
We will say that $f(z) \in \Hahn_{R}\setminus \{0\}$ is tangent to the identity if $\val(f(z))=0$ and $\cld(f(z)) = 1$. 

\subsection{Structure of difference ring on the ring of Hahn series $\Hahn_{R}$}

One can endow $\Hahn_{R}$ with the ring automorphism 
\begin{eqnarray}\label{phisurHR}
 \malop{\ellmahl} : \ \ \ \ \ \ \ \ \ \Hahn_{R} \ \ \ \ \ \ &\rightarrow& \ \ \ \Hahn_{R}\\
f(z)=\sum _{{\gamma\in \Q }}f_{\gamma}z^{\gamma} &\mapsto& \sum _{{\gamma\in \Q }}f_{\gamma}z^{\ellmahl\gamma}.\nonumber
\end{eqnarray}
We obtain in this way a structure of difference field on $\Hahn_{R}$.

\subsection{The differential-difference ring $\Hahn_{\C(\lambda)}$}

We will make an essential use of the field $\Hahn_{\C(\lambda)}$ of Hahn series with coefficients in the field of rational functions $\C(\lambda)$ and with value group $\Q$. 

One can endow $\Hahn_{\C(\lambda)}$ with the field automorphism 
\begin{equation}\label{phisurHClambda}
 \malop{\ellmahl} : \Hahn_{\C(\lambda)}\rightarrow \Hahn_{\C(\lambda)}
 \end{equation}
defined by \eqref{phisurHR}
and with the derivation 
\begin{eqnarray}\label{partialsurHClambda}
 \partial_{\lambda} =  \frac{\partial \  }{\partial \lambda} : \ \ \ \ \ \ \ \ \ \Hahn_{\C(\lambda)} \ \ \ \ \ \ \ &\rightarrow &  \ \ \  \Hahn_{\C(\lambda)}\\
 f(\lambda,z)=\sum _{{\gamma\in \Q }}f_{\gamma}(\lambda)z^{\gamma} &\mapsto& \sum _{{\gamma\in \Q }}\frac{\mathrm{d}   f_{\gamma}}{\mathrm{d} \lambda} (\lambda)z^{\gamma}.\nonumber
\end{eqnarray}
We obtain in this way a structure of differential-difference field on $\Hahn_{\C(\lambda)}$. We emphasize that $\partial_{\lambda}$ and $\malop{\ellmahl}$ commute on $\Hahn_{\C(\lambda)}$:
$$
\partial_{\lambda} \circ \malop{\ellmahl} = \malop{\ellmahl} \circ \partial_{\lambda}. 
$$

The role played by the differential-difference field $\Hahn_{\C(\lambda)}$ in the present paper is similar to the role played by the differential field $\C(\lambda)((z))$  in the classical Frobenius method outlined in Section \ref{classical frob method}. 

\subsection{The differential-difference ring $\Rla$} We shall now introduce a differential-difference ring $\Rla$ that will play the same role as the differential ring 
$$
\C(\lambda)((z))[z^{\lambda},\partial_{\lambda} (z^{\lambda}),\partial_{\lambda}^{2}(z^{\lambda}),\ldots]
$$
(equipped with the derivatives $\partial_{z}$ and $\partial_{\lambda}$) in the classical Frobenius method. 

We consider the differential ring of differential polynomials  
$$
\Rla=\Hahn_{\C(\lambda)}[X,\partial_{\lambda}X,\partial_{\lambda}^{2}X,\ldots]. 
$$  
As a ring, this is simply the ring of polynomials with coefficients in $\Hahn_{\C(\lambda)}$ and in the indeterminates $X,\partial_{\lambda}X,\partial_{\lambda}^{2}X,\ldots$
This ring is equipped with the unique derivation
$$
\partial_{\lambda}: \Rla \rightarrow \Rla
$$
extending \eqref{partialsurHClambda} and such that 
$
\partial_{\lambda}(\partial_{\lambda}^{i}X)=\partial_{\lambda}^{i+1}X.
$
%Since $\partial_{\lambda}$ and $\malop{\ellmahl}$ commute on $\Hahn_{\C[\lambda]}$, 
We also equip $\Rla$ with the unique ring automorphism 
$$
\malop{\ellmahl}: \Rla \rightarrow \Rla
$$ 
extending \eqref{phisurHClambda} and such that, 
 for all $i \geq 0$, $\malop{\ellmahl}(\partial_{\lambda}^{i}X)=\partial_{\lambda}^{i}(\lambda X)$. 
This turns $\Rla$ into a differential-difference ring extension of $\Hahn_{\C(\lambda)}$. We emphasize that $\partial_{\lambda}$ and $\malop{\ellmahl}$ commute on $\Rla$:
$$
\partial_{\lambda} \circ \malop{\ellmahl} = \malop{\ellmahl} \circ \partial_{\lambda}. 
$$

It  will be convenient to introduce the following notations : 
\begin{itemize}
 \item $
e_{\lambda} = X \in \Rla
$;
\item for all $i \geq 0$, 
$
\logm_{\lambda,i} = \frac{\partial_{\lambda}^{i}e_{\lambda}}{i!}=\frac{\partial_{\lambda}^{i}X}{i!} \in \Rla; 
$
\end{itemize}
so that $\Rla$ can be rewritten as follows :
\begin{eqnarray*}
\Rla &=& \Hahn_{\C(\lambda)}[e_{\lambda},\partial_{\lambda}(e_{\lambda}),\partial_{\lambda}^{2}(e_{\lambda}),\ldots] \\
&=&  \Hahn_{\C(\lambda)}[e_{\lambda},\logm_{\lambda,1},\logm_{\lambda,2},\ldots]. 
\end{eqnarray*}
It is of fundamental importance in what follows to keep in mind that  $e_{\lambda}$ satisfies the following equation 
\begin{equation}\label{eq de elambda}
\malop{\ellmahl}(e_{\lambda})=\lambda e_{\lambda}.  
\end{equation}
More generally, for all $i \geq 0$, we have 
$$
\malop{\ellmahl}(\logm_{\lambda,i}) 
= 
\lambda \logm_{\lambda,i} +   \logm_{\lambda,i-1}
$$ 
where $\logm_{\lambda,i-1}=\logm_{\lambda,-1}=0$ for $i=0$. 

\begin{remark}
 The difference equation \eqref{eq de elambda} satisfied by $e_{\lambda}$ has to be compared with the differential equation 
$$
z\partial_{z}(z^{\lambda})=\lambda z^{\lambda}
$$ satisfied by $z^{\lambda}$; this  $e_{\lambda}$ will play in the present paper a role similar to that of $z^{\lambda}$ in Section~\ref{classical frob method}. More generally, the  $\logm_{\lambda,i}$ will play a role similar to that of the $\partial_{\lambda}^{i}(z^{\lambda})$ in Section \ref{classical frob method}. 
\end{remark}

\subsection{The difference ring $\Rspe$}\label{sec : def Rspe}

Throughout this paper, we denote by $\Rspe$ a difference ring extension of $\Hahn$ such that, for any $c \in \C^{\times}$ and any integer $i\geq 0$, there exists a nonzero $\logm_{c,i} \in \Rspe$ such that 
\begin{equation}\label{eq logmcj}
 \malop{\ellmahl}(\logm_{c,i} )=c\logm_{c,i} +\logm_{c,i-1}, 
\end{equation}
where $\logm_{c,i-1}=\logm_{c,-1}=0$ for $i=0$. For $i=0$, we use the notation 
$$
e_{c}=\logm_{c,0}.
$$ 
Of course, equation \eqref{eq logmcj} gives 
$$
\malop{\ellmahl}(e_{c})=c e_{c}.
$$
The solutions produced by the Frobenius method presented in this paper will belong to $\Rspe$. 

\begin{remark}\label{rem: il suffit l}
 It is sufficient to require that $\Rspe$ satisfies the following properties:
 \begin{itemize}
 \item for any $c \in \C^{\times}$, there exists $e_{c} \in \Rspe$ which is not a zero divisor such that $\malop{\ellmahl} (e_{c}) = c e_{c}$; 
 \item there exists $\logm \in \Rspe$ such that $\malop{\ellmahl}(\logm)=\logm+1$. 
\end{itemize}
 Indeed, a straightforward calculation shows that 
$$
\logm_{c,i} = c^{-i}e_{c} \binom{\logm}{i} 
$$
is nonzero and satisfies \eqref{eq logmcj}. 
\end{remark}

A possible choice is the polynomial ring 
$$
\Rspe=\Hahn[(X_{c})_{c \in \C^{\times}},Y]
$$
endowed with the unique automorphism 
$$
\malop{\ellmahl} : \Rspe \rightarrow \Rspe
$$
extending \eqref{phisurHR} and such that $\malop{\ellmahl} (X_{c})=cX_{c}$ and $\malop{\ellmahl}(Y)=Y+1$. Of course, $e_{c}=X_{c}$ and $\logm=Y$ have the properties required in Remark \ref{rem: il suffit l}. 

However, we emphasize that the specific choices of $e_{c}$ and $\logm_{c,i}$ is not important; the important thing are the functional equations \eqref{eq logmcj} they satisfy.

\subsection{Evaluating elements of $\Rla$ at $\lambda = c \in \C^{\times}$} As in the classical Frobenius method (see Section \ref{classical frob method}), we will need to ``specialize elements of $\Rla$ at $\lambda=c \in \C^{\times}$''. Here is what we mean by ``specialize'' in this context.

We let $\C[\lambda]_{(\lambda-c)}$ be the ring of rational fractions regular at $c \in \C^{\times}$. Then, 
$\Hahn_{\C[\lambda]_{(\lambda-c)}}
$ 
is a differential-difference subring of $\Hahn_{\C(\lambda)}$ 
and  
\begin{equation}\label{Rlamc}
\Rspe_{\lambda,c} = \Hahn_{\C[\lambda]_{(\lambda-c)}}[e_{\lambda},\logm_{\lambda,1},\logm_{\lambda,2},\ldots]
\end{equation}
is a differential-difference subring of $\Rla$. 
There is a unique morphism of rings 
$$
\ev{c} : \Rspe_{\lambda,c} \rightarrow \Rspe
$$ 
such that, for all $f(\lambda,z) \in \Hahn_{\C[\lambda]_{(\lambda-c)}}$,
$$
\ev{c}(f(\lambda,z))=f(c,z)
$$
and, for all $i \geq 0$,
$$
 \ev{c}(\logm_{\lambda,i})=\logm_{c,i}.
$$
This is a morphism of difference rings. In the context of this paper, ``specilizing at $\lambda=c$'' means taking the image by $\ev{c}$.

\section{Newton polygons, slopes and exponents}\label{sec: newt pol}

In this section, we consider a Mahler equation 
\begin{equation}\label{eq mahl slopes}
 a_{n}(z) y(z^{\ellmahl^{n}}) + a_{n-1}(z) y(z^{\ellmahl^{n-1}}) + \cdots + a_{0}(z) y(z) = 0
\end{equation}
with coefficients $a_{0}(z),\ldots,a_{n}(z) \in \Hahn$. We can and will assume that $a_{0}(z)a_{n}(z)\neq 0$.

The aim of this Section is to attach to \eqref{eq mahl slopes} certain slopes and to attach to each slope certain exponents. 

\subsection{Mahler equations and Mahler operators}

The Mahler equation \eqref{eq mahl slopes} can be written as follows: 
$$
L (y(z)) = 0
$$
where 
\begin{equation}\label{mahl op assoc}
 L=a_{n}(z)\malop{\ellmahl}^{n} + a_{n-1}(z) \malop{\ellmahl}^{n-1} + \cdots + a_{0}(z). 
\end{equation}
This is an element of the Ore algebra 
$$
\mathcal{D}_{\Hahn}=\Hahn\langle \malop{\ellmahl},\malop{\ellmahl}^{-1} \rangle
$$ 
of noncommutative Laurent polynomials with coefficients in $\Hahn$ such that, for all $f\in \Hahn$, 
$
\malop{\ellmahl} f(z) = \malop{\ellmahl}(f(z)) \malop{\ellmahl}.
$ 
An element of $\mathcal{D}_{\Hahn}$ will be called a Mahler operator (with coefficients in $\Hahn$). The Mahler operator $L$ given by \eqref{mahl op assoc} will be called the Mahler operator associated with \eqref{eq mahl slopes}.

\subsection{Newton polygons}\label{subsec:slopes}
We define the Newton polygon $\mathcal N(L)$ of $L$ 
as the convex hull in $\R^2$ of 
$$
\{(\ellmahl^{i},j) \ \vert \ i,j \in \Z, \ j \geq \val(a_{i}(z))\} \subset \R^{2}
$$
where $\val : \Hahn \rightarrow \Q \cup \{+\infty\}$ denotes the $z$-adic valuation (this is for instance the Newton polygon used in \cite{CompSolMahlEq}). 

\subsection{Slopes}\label{sec : def slopes} Let 
\begin{multline*}
 v_{0}=(\ellmahl^{\alpha_{0}},\val(a_{\alpha_{0}}(z)))=(p^{0},\val(a_{0}(z))),\ldots\\
\ldots,v_{i}=(\ellmahl^{\alpha_{i}},\val(a_{\alpha_{i}}(z))),\ldots\\
\ldots,v_{k}=(\ellmahl^{\alpha_{k}},\val(a_{\alpha_{k}}(z)))=(\ellmahl^{n},\val(a_{n}(z)))
\end{multline*}
be the vertices, ordered by increasing abscissa, of the polygon $\mathcal N(L)$. This polygon is delimited by  two vertical half lines and by $k$ vectors 
$$
w_{1}=v_{1}-v_{0}, \ldots, w_{k}=v_{k}-v_{k-1} \in \Z_{>0} \times \Q
$$ 
having pairwise distinct slopes denoted by 
$$
\mu_{1} < \cdots < \mu_{k}
$$ 
and called the slopes of $L$. The set of slopes of $L$ will be denoted by $\mathcal{S}(L)$. For any $i\in \{1,\ldots,k\}$, the integer 
$$
r(\mu_{i},L)= \alpha_{i}-\alpha_{i-1}
$$ 
is called the multiplicity of $\mu_{i}$ as a slope of $L$; in what follows, we will also use the shorthand notation 
$$
r_{i}=r(\mu_{i},L).
$$ 

\begin{example}\label{example order 2}
The Newton polygon of the Mahler operator 
\begin{eqnarray}\label{eq:un exemple bis}
 L&=&(\malop{\ellmahl}-z^{\nu}) h(z)^{-1} (\malop{\ellmahl}-1)\\
 &=& \frac{1}{1+ z^{-\frac{\ellmahl\nu}{\ellmahl-1}}}\malop{\ellmahl}^{2} - \left(\frac{1}{1+ z^{-\frac{\ellmahl\nu}{\ellmahl-1}}}+z^{\nu}\right) \malop{\ellmahl} + \frac{z^{\nu}}{1+ z^{-\frac{\nu}{\ellmahl-1}}} \nonumber
\end{eqnarray}
with $\nu \in \Q_{<0}$ and $h(z)=1+z^{-\frac{\nu}{\ellmahl-1}}$ is represented in Figure \ref{fig:Newt bis}.

This operator has two slopes, namely  
$$
\mu_{1}=0 < \mu_{2}=-\frac{\nu}{(\ellmahl-1)\ellmahl},
$$ 
with multiplicities $r_{1}=r_{2}=1$.
\end{example}

 \begin{figure}
  \begin{tikzpicture}[inner sep=0,x=35pt,y=35pt,font=\footnotesize]
  \fill[fill=black!5!white] (1,2)--(1,-1)--(4,-1)--(8,0)--(8,2);
   \filldraw (0,0) circle (1pt);
   \filldraw (1,0) circle (1pt);
   \filldraw (4,0) circle (1pt);
   \filldraw (8,0) circle (1pt);
   \filldraw (0,-1) circle (1pt);
   \filldraw (0,2) circle (1pt);
   \filldraw (0,1) circle (1pt);
   \filldraw (0,2) circle (1pt);
   \node at (-0.15,0) {$0$};
   \node at (-0.15,1) {$1$};
   \node at (-0.15,2) {$2$};
   \node at (-0.25,-1) {$\nu$};
   \node at (0.9,-0.2) {$1$};
  % \node at (2.0,-0.2) {$2$};
   % \node at (3.0,-0.2) {$3$};
   \node at (4.0,-0.2) {$\ellmahl$};
   %\node at (5.0,-0.2) {$5$};
   %\node at (6.0,-0.2) {$6$};
   %\node at (7.0,-0.2) {$7$};
   \node at (8.0,-0.2) {$\ellmahl^{2}$};
   \filldraw (1,-1) circle (2pt);
   %\filldraw (2,1) circle (2pt);
   \filldraw (4,-1) circle (2pt);
   \filldraw (8,0) circle (2pt);
   \node at (0.75,-0.85) {$v_{0}$};
   \node at (2.4,1) {};
   \node at (4.2,-1.15) {$v_{1}$};
   \node at (8.2,0.15) {$v_{2}$};
  % \draw[-, thick]  (1,2) -- (1,-1) -- (4,-1) -- (8,0) -- (8,2);
  \draw[-, thick]  (1,2) -- (1,-1) -- (1.5,-1) ;
   \draw[dashed, thick]   (1.5,-1) -- (3.5,-1) ;
  % \draw[-, thick]   (3.5,-1) -- (4,-1) -- (8,0) -- (8,2);
    \draw[-, thick]   (3.5,-1) -- (4,-1)-- (4.5,-0.875) ;
     \draw[dashed, thick]   (1.5,-1) -- (3.5,-1) ;
    \draw[dashed, thick]   (4.5,-0.875)-- (7.5,-0.125) ;
    \draw[-, thick]    (7.5,-0.125) -- (8,0) ;
    \draw[-, thick]    (8,0) -- (8,2) ;
%   \node at (2.4,-0.7) {$L_2$};
%   \node at (4.75,-0.6) {$L_3$};
 \draw[-] (0,0) -- (1.5,0) node[below=5pt] {};
 \draw[dashed] (1.5,0) -- (3.5,0) node[below=5pt] {};
 \draw[-] (3.5,0) -- (4.5,0) node[below=5pt] {};
  \draw[dashed] (4.5,0) -- (7.5,0) node[below=5pt] {};
   \draw[->] (7.5,0) -- (9,0) node[below=5pt] {};
   \draw[->] (0,-0.2) -- (0,2.3) node[left=10pt] {};
   \draw[dashed] (0,-0.8) -- (0,-0.2) node[below=5pt] {};
    \draw[-] (0,-1.2) -- (0,-0.8) node[below=5pt] {};
   \node[font=\normalsize] at (6.0,1.5) {
   };
  \end{tikzpicture}
  \caption{The Newton polygon of the Mahler operator  \eqref{eq:un exemple bis}. 
The $\bullet$ represent the points $(\ellmahl^{i},\val(a_{i}))$ for $i \in \{0,1,2\}$.  }
\label{fig:Newt bis}
  \end{figure}
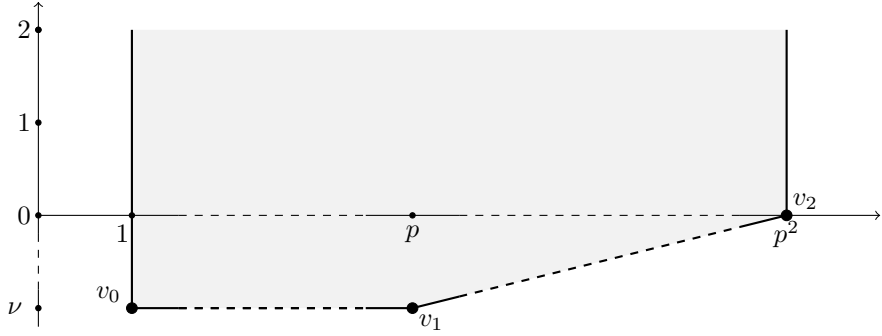

\begin{remark}
We have extracted two informations from $\mathcal{N}(L)$ : its slopes and their multiplicities. Of course, $\mathcal{N}(L)$ is characterized by these data up to vertical translation. 
\end{remark}

For any $\mu \in \Q$, we consider $$
\theta_{\mu}(z)=z^{\frac{\mu}{\ellmahl-1}}.
$$ 
It satisfies  
$$
\malop{\ellmahl} (\theta_{\mu}(z)) = z^{\mu} \theta_{\mu}(z).
$$ 
We set   
\begin{multline*}
 L^{[\theta_{\mu}(z)]}=\theta_{\mu}(z)^{-1}L\theta_{\mu}(z) \\ 
 =\sum_{i=0}^{n} z^{(1+\ellmahl+\cdots+\ellmahl^{i-1})\mu} a_i(z) \malop{\ellmahl}^i
=
\sum_{i=0}^{n} z^{\frac{\ellmahl^{i}-1}{\ellmahl-1}\mu} a_i(z) \malop{\ellmahl}^i.
\end{multline*}

\begin{prop}\label{prop: slopes l theta mu}
The slopes of $L^{[\theta_{\mu}(z)]}$ are 
$$
\mu_{1} + \frac{\mu}{\ellmahl-1} < \cdots < \mu_{k}+\frac{\mu}{\ellmahl-1}
$$ 
with respective multiplicities $r_{1},\ldots,r_{k}$. 
%are respectively given by given by 
%$$
%r\left(\mu_{j}+\frac{\mu}{\ellmahl-1} ,L^{[\theta_{\mu_{j}}]}\right)
%=r(\mu_{j},L). 
%$$
\end{prop}

\begin{proof}
The Newton polygon $\mathcal N(L^{[\theta_{\mu}(z)]})$ is the convex hull of  
\begin{equation*}
 \left\{(\ellmahl^{i},j) \ \vert \ i,j \in \Z, \ j \geq \val(a_{i}(z))+\mu \frac{\ellmahl^{i}-1}{\ellmahl-1}\right\}. 
\end{equation*}
Its vertices are 
\begin{multline*}
 v'_{0}=v_{0}+(0,\mu \frac{\ellmahl^{\alpha_{0}}-1}{\ellmahl-1}),\ldots \\
 \ldots, v_{i}'=v_{i}+(0,\mu \frac{\ellmahl^{\alpha_{i}}-1}{\ellmahl-1}),\ldots \\\ldots,v_{k}'=v_{k}+(0,\mu \frac{\ellmahl^{\alpha_{k}}-1}{\ellmahl-1}),
\end{multline*}
and, hence, $\mathcal N(L^{[\theta_{\mu}]})$
is delimited by  two vertical half lines and by the $k$ vectors 
\begin{multline*}
w_{1}'=v_{1}'-v_{0}'=w_{1}+(0,\mu \frac{\ellmahl^{\alpha_{1}}-\ellmahl^{\alpha_{0}}}{\ellmahl-1}), \ldots \\
 \ldots,
  w_{i}'=v_{i}'-v_{i-1}'=w_{i} +(0,\mu \frac{\ellmahl^{\alpha_{i}}-\ellmahl^{\alpha_{i-1}}}{\ellmahl-1}),
  \ldots\\
\ldots , w_{k}'=v_{k}'-v_{k-1}'=w_{k} +(0,\mu \frac{\ellmahl^{\alpha_{k}}-\ellmahl^{\alpha_{k-1}}}{\ellmahl-1}).
\end{multline*} 
\end{proof}

In particular, the $j$-th slope of $L^{[\theta_{-(\ellmahl-1)\mu_{j}}(z)]}$ is equal to $0$. This will be used in the next section in order to attach to the slope $\mu_{j}$ a certain characteristic polynomial and certain exponents.

\subsection{Characteristic equations and exponents}\label{subsec:exponents}

Consider a slope  $\mu_{j}$ of $L$ and set 
$$
L^{[\theta_{-(\ellmahl-1)\mu_{j}}(z)]}=\sum_{i=0}^{n} b_{i}(z) \malop{\ellmahl}^{i}.
$$ 
Let 
$$\val(L^{[\theta_{-(\ellmahl-1)\mu_{j}}]}) = \min \{\val(b_{0}(z)),\ldots,\val(b_{n}(z))\}.
$$ 
The characteristic polynomial $\chi(\mu_{j},L;X)$ associated 
to the slope $\mu_{j}$ of $L$ is  
\begin{eqnarray*}
 \chi(\mu_{j},L;X)&=&\left(z^{-\val(L^{[\theta_{-(\ellmahl-1)\mu_{j}}(z)]})} \sum_{i=0}^{n} b_{i}(z)X^{i}\right)_{\vert z=0} \\
 &=&\left(z^{-\val(L^{[\theta_{-(\ellmahl-1)\mu_{j}}(z)]})}\sum_{i=\alpha_{j-1}}^{\alpha_{j}} b_{i}(z) X^{i} \right)_{\vert z=0}. 
\end{eqnarray*}

In what follows, we will see these characteristic polynomials as defined modulo a multiplicative factor in $\C^{\times} X^{\Z}$; in particular, the equalities involving the characteristic polynomials have to be interpreted modulo $\C^{\times} X^{\Z}$.  
The characteristic polynomial $\chi(\mu_{j},L;X)$ has $r_{j}$ roots (counted with multiplicities) in $\C^{\times}$  called the exponents of $L$ attached to the slope $\mu_{j}$. The multiplicity of such an exponent is its multiplicity as a root of $\chi(\mu_{j},L;X)$.
 
 \begin{example}\label{ex:2}
 We have already seen that the Mahler operator $L$ given by  
\eqref{eq:un exemple bis} has two slopes 
$$
\mu_{1}=0 < \mu_{2}=-\frac{\nu}{(\ellmahl-1)\ellmahl},
$$ 
with multiplicities $r_{1}=r_{2}=1$.

Let us compute the characteristic polynomial of the Mahler operator~\eqref{eq:un exemple bis} associated to the slope $\mu_{1}$. We have $$
L^{[\theta_{-(\ellmahl-1)\mu_{1}}(z)]}=L=\frac{1}{1+ z^{-\frac{\ellmahl\nu}{\ellmahl-1}}}\malop{\ellmahl}^{2} - \left(\frac{1}{1+ z^{-\frac{\ellmahl\nu}{\ellmahl-1}}}+z^{\nu}\right) \malop{\ellmahl} + \frac{z^{\nu}}{1+ z^{-\frac{\nu}{\ellmahl-1}}} 
$$ 
and 
$$\val(L^{[\theta_{-(\ellmahl-1)\mu_{1}}(z)]}) = \nu.
$$ 
Thus, the characteristic polynomial $\chi(\mu_{1},L;X)$ associated 
to the slope $\mu_{1}$ of $L$ is  
\begin{multline*}
 \chi(\mu_{1},L;X)=\\ 
 \left(z^{-\nu}\left(  \frac{1}{1+ z^{-\frac{\ellmahl\nu}{\ellmahl-1}}}X^{2} - \left(\frac{1}{1+ z^{-\frac{\ellmahl\nu}{\ellmahl-1}}}+z^{\nu}\right) X + \frac{z^{\nu}}{1+ z^{-\frac{\nu}{\ellmahl-1}}} \right) \right)_{\vert z=0} \\
 = -X+1.
\end{multline*}
So, $1$ is the unique exponent of $L$ attached to the slope $\mu_{1}$ and it has multiplicity $1$.

Let us compute the characteristic polynomial of the Mahler operator~\eqref{eq:un exemple bis} associated to the slope $\mu_{2}$. We have 
\begin{eqnarray}
 L^{[\theta_{-(\ellmahl-1)\mu_{2}}(z)]}
 &=& \frac{z^{(\ellmahl+1)\frac{\nu}{\ellmahl}}}{1+ z^{-\frac{\ellmahl\nu}{\ellmahl-1}}}\malop{\ellmahl}^{2} - z^{\frac{\nu}{\ellmahl}}\left(\frac{1}{1+ z^{-\frac{\ellmahl\nu}{\ellmahl-1}}}+z^{\nu}\right) \malop{\ellmahl} + \frac{z^{\nu}}{1+ z^{-\frac{\nu}{\ellmahl-1}}} \nonumber
\end{eqnarray}
and 
$$\val(L^{[\theta_{-(\ellmahl-1)\mu_{2}}(z)]}) = (\ellmahl+1)\frac{\nu}{\ellmahl}.
$$ 
Thus, the characteristic polynomial $\chi(\mu_{2},L;X)$ associated 
to the slope $\mu_{2}$ of $L$ is  
\begin{multline*}
 \chi(\mu_{2},L;X)=\\ 
 \left(z^{-(\ellmahl+1)\frac{\nu}{\ellmahl}}\left(   \frac{z^{(\ellmahl+1)\frac{\nu}{\ellmahl}}}{1+ z^{-\frac{\ellmahl\nu}{\ellmahl-1}}}X^{2} - z^{\frac{\nu}{\ellmahl}}\left(\frac{1}{1+ z^{-\frac{\ellmahl\nu}{\ellmahl-1}}}+z^{\nu}\right) X + \frac{z^{\nu}}{1+ z^{-\frac{\nu}{\ellmahl-1}}} \right) \right)_{\vert z=0} \\
 = X^{2}-X. 
\end{multline*}
So, $1$ is the unique exponent of $L$ attached to the slope $\mu_{2}$ and it has multiplicity $1$.
\end{example}

 \section{An extension of Frobenius method to Mahler equations}\label{sec: statement frob}
 
 We consider a Mahler equation 
 $$
 a_{n}(z) y(z^{\ellmahl^{n}}) + a_{n-1}(z) y(z^{\ellmahl^{n-1}}) + \cdots + a_{0}(z) y(z) = 0
 $$
 with $a_{0}(z),\ldots,a_{n}(z) \in \Hahn$ and $a_{0}(z)a_{n}(z)\neq 0$.
It can be rewritten as 
 \begin{equation}\label{eq mahl slopes sec meth frob}
 L (y(z))=0
\end{equation}
where 
\begin{equation*}\label{op mahl descr frob meth}
L= a_{n}(z) \malop{\ellmahl}^{n} + a_{n-1}(z) \malop{\ellmahl}^{n-1} + \cdots + a_{0}(z) \in \mathcal{D}_{\Hahn}.
\end{equation*}

We denote by 
$$
\mu_{1} < \cdots < \mu_{k}
$$  
the slopes of $L$ and by $r_{1},\ldots,r_{k}$ their respective multiplicites. The multiplicity of an exponent $c$ of $L$ attached to the slope $\mu_{j}$ will be denoted by $m_{c,j} \in \Z_{\geq 1}$. If $c$ is a nonzero complex number that is not an exponent attached to the slope $\mu_{j}$, we set $m_{c,j}=0$. We refer the reader to Section~\ref{sec: newt pol} for details about the notions of slopes and  exponents. 

The extension of Frobenius method to Mahler equations presented below attaches $m_{c,j}$ solutions of \eqref{eq mahl slopes sec meth frob} to any exponent $c$  associated to the slope $\mu_{j}$ of $L$ as follows. We will prove that 
\begin{itemize}
 \item there exists a unique $g_{c,j}(\lambda,z)\in\Hahnlamb$ such that 
\begin{equation}\label{eq frob bis bis bis bis bis}
 L(g_{c,j}(\lambda,z)e_{\lambda}) = z^{\val(a_{0}(z))-\frac{\nu_{j}}{\ellmahl-1}}(\lambda-c)^{s_{c,j}+m_{c,j}} e_{\lambda}
\end{equation}
where 
$$
s_{c,j} = m_{c,1}+\cdots+m_{c,j-1}
$$
and 
$$
\nu_{j} =(\ellmahl-1) (\ellmahl^{r_{1}+\cdots+r_{j-1}}(\mu_{j}-\mu_{j-1})+\cdots+\ellmahl^{r_{1}}(\mu_{2}-\mu_{1})+\mu_{1}); 
$$
\item the $z$-adic valuation of $g_{c,j}(\lambda,z)$ is equal to $-\mu_{j}$;
 \item  the coefficients of $g_{c,j}(\lambda,z)$ have no pole at $\lambda=c$; 
 \item  the specializations 
 \begin{eqnarray*}
 &\ev{c}(\partial_{\lambda}^{s_{c,j}} (g_{c,j}(\lambda,z)e_{\lambda})),& \\ 
 &\ev{c}(\partial_{\lambda}^{s_{c,j}+1} (g_{c,j}(\lambda,z)e_{\lambda})),& \\
  &\ldots& \\
& \ev{c}(\partial_{\lambda}^{s_{c,j}+m_{c,j}-1} (g_{c,j}(\lambda,z)e_{\lambda}))&
\end{eqnarray*}    
are $m_{c,j}$ solutions of \eqref{eq mahl slopes sec meth frob}. 
\end{itemize}
For the justifications of the first three properties above, see Proposition \ref{resol homogene pour frob} in Section \ref{sec:frob method}. For the last one, see Proposition \ref{prop: les ycjm sol} in Section \ref{sec: frob meth last just}. 

We can now state the main result of this paper. 

\begin{theo}\label{theo:theo pas intro}
 We have attached to any slope $\mu_{j}$,  to any exponent $c$ attached to the slope $\mu_{j}$ and to any $m \in \{0,\ldots,m_{c,j}-1\}$, a solution 
$$
y_{c,j,m}=\ev{c}(\partial_{\lambda}^{s_{c,j}+m} (g_{c,j}(\lambda,z)e_{\lambda}))
$$ 
of~\eqref{eq mahl slopes sec meth frob}.   
We obtain in this way a family of $n$ $\C$-linearly independent solutions of~\eqref{eq mahl slopes sec meth frob}. 
\end{theo}

The proof of this result will be given in Section \ref{sec: frob meth last just}. 

Note that 
\begin{eqnarray*}
 y_{c,j,m} & \in & \operatorname{Span}_{\Hahn} (\ev{c}(\partial_{\lambda}^{0} (e_{\lambda})),\ev{c}(\partial_{\lambda}^{1} (e_{\lambda})),\ldots \\
&& \ \ \ \ \ \ \ \ \ \ \ \ \ \ \ \ \ \ \ \ \ \ \ \ \ \ \ \ \ \ \ \ \ \ \ \ \ \ \ \ \ldots ,\ev{c}(\partial_{\lambda}^{s_{c,j}+m_{c,j}-1} (e_{\lambda}))) \\ 
 &=& \operatorname{Span}_{\Hahn} (e_{c},\logm_{c,1},\ldots,\logm_{c,s_{c,j}+m_{c,j}-1}). 
\end{eqnarray*}

\begin{remark}
The solutions attached to $c$ involve an $\logm_{c,k}$ with $k \geq 1$ if $c$ has multiplicity $\geq 2$ as an exponent of $L$ attached to a slope $\mu_{j}$ ({\it i.e.}, $m_{c,j} \geq 2$). But, this is not the only case : such an $\logm_{c,k}$ may also appear if $c$ is an exponent of $L$ attached to two distinct slopes $\mu_{i}$ and $\mu_{j}$ ({\it i.e.}, $s_{c,j} \geq 1$ if $j>i$). 
\end{remark}

\begin{remark} 
The relevant fact with the term $z^{-\frac{\nu_{j}}{\ellmahl-1}}$ involved in the right-hand side of \eqref{eq frob bis bis bis bis bis} is its $z$-adic valuation equals to $-\frac{\nu_{j}}{\ellmahl-1}$~: it guarantees that $g_{c,j}(\lambda,z)$ has $z$-adic valuation $-\mu_{j}$ and, hence, is ``associated'' with the slope $\mu_{j}$. 
\end{remark}

\section{Frobenius method for Mahler equations : an example}\label{sec: main example}

In this section, we apply the method described in Section \ref{sec: statement frob} to an operator allowing explicit computations, namely to the operator considered in Example \ref{example order 2} given by
\begin{eqnarray*}
 L
 &=&(\malop{\ellmahl}-z^{\nu}) h(z)^{-1} (\malop{\ellmahl}-1)\\
 &=& \frac{1}{1+ z^{-\frac{\ellmahl\nu}{\ellmahl-1}}}\malop{\ellmahl}^{2} - \left(\frac{1}{1+ z^{-\frac{\ellmahl\nu}{\ellmahl-1}}}+z^{\nu}\right) \malop{\ellmahl} + \frac{z^{\nu}}{1+ z^{-\frac{\nu}{\ellmahl-1}}}\\
&=& a_{2}(z)\malop{\ellmahl}^{2}+a_{1}(z)\malop{\ellmahl}+a_{0}(z)
\end{eqnarray*}
where $\nu \in \Q_{<0}$ and $h(z)=1+z^{-\frac{\nu}{\ellmahl-1}}$.

In what follows, we use the notations $\mu_{i}$, $r_{i}$, $\nu_{j}$, {\it etc}, of Section \ref{sec: statement frob}. 

We have seen in Example \ref{example order 2} that $L$ has two slopes, namely  
$$
\mu_{1}=0 < \mu_{2}=-\frac{\nu}{(\ellmahl-1)\ellmahl},
$$ 
with multiplicities 
$$r_{1}=r_{2}=1.$$ 
So,  
$$\nu_{1}=(\ellmahl-1)\mu_{1}=0, \ \ \nu_{2}=(\ellmahl-1)(p^{r_{1}}(\mu_{2}-\mu_{1})+\mu_{1})=-\nu.
$$ 
Moreover, we have seen in Example~\ref{ex:2} that $1$ is the unique exponent of $L$ and that the multiplicites $m_{1,1}$ and $m_{1,2}$ of $1$ as an exponent of $L$ attached to the slopes $\mu_{1}$ and $\mu_{2}$ respectively are  
$$
m_{1,1}=m_{1,2}=1.
$$ 
Thus 
$$
s_{1,1}=0, \ \ \ s_{1,2}=1.
$$

The Frobenius method described in Section \ref{sec: statement frob} relies on the equations 
\begin{eqnarray}\label{eq:first eq}
 L( g_{1}(\lambda,z)e_{\lambda})&=& z^{\val(a_{0}(z))} \theta_{-\nu_{1}}(z)(\lambda-1)^{m_{1,1}} e_{\lambda} \\ 
 &=&z^{\nu}  (\lambda-1)e_{\lambda} \nonumber
\end{eqnarray}
and 
\begin{eqnarray}\label{eq:second eq}
 L( g_{2}(\lambda,z)e_{\lambda})&=& z^{\val(a_{0}(z))} \theta_{-\nu_{2}}(z) (\lambda-1)^{m_{1,1}+m_{1,2}} e_{\lambda} \\ 
 &=&z^{\nu} \theta_{\nu}(z) (\lambda-1)^{2} e_{\lambda}.  \nonumber
\end{eqnarray}
We shall now give explicit formulas for the unique $ g_{1}(\lambda,z), g_{2}(\lambda,z) \in \Hahn_{\C(\lambda)}$ satisfying the above equations. In this purpose, we will freely use Lemma \ref{lem:sol ordre un inhom mu z c un} stated and proven in Section \ref{sec: prel res first just} below. 

Let us start with equation \eqref{eq:first eq}. This equation can be written as follows
$$(\lambda \malop{\ellmahl}-z^{\nu})  h(z)^{-1}(\lambda \malop{\ellmahl}-1)(g_{1}(\lambda,z))=z^{\nu} (\lambda-1).$$
Multiplying by $\theta_{\nu}(z)^{-1}z^{-\nu}$, we see that the latter equation is equivalent to 
$$
(\lambda \malop{\ellmahl}-1)(\theta_{\nu}(z)^{-1}h(z)^{-1}(\lambda \malop{\ellmahl}-1)(g_{1}(\lambda,z))) =\theta_{\nu}(z)^{-1}(\lambda-1). 
$$
Since the $z$-adic valuation of the right hand-side of the latter equality is positive, Lemma \ref{lem:sol ordre un inhom mu z c un} ensures that the latter equation is equivalent to 
$$
\theta_{\nu}^{-1}h(z)^{-1}(\lambda \malop{\ellmahl}-1)(g_{1}(\lambda,z)) = -(\lambda-1) \sum_{k \geq 0} \lambda^{k}\malop{\ellmahl}^{k} (\theta_{\nu}(z)^{-1}), 
$$
{\it i.e.}, to 
\begin{multline*}
(\lambda \malop{\ellmahl}-1)(g_{1}(\lambda,z))= -(\lambda-1) (1+z^{-\frac{\nu}{\ellmahl-1}})\sum_{k \geq 0} \lambda^{k}z^{\frac{\nu(1-\ellmahl^{k})}{\ellmahl-1}} \\ 
=  -(\lambda-1) 
-(\lambda-1) \sum_{k \geq 1} \lambda^{k}z^{\frac{\nu(1-\ellmahl^{k})}{\ellmahl-1}}
-(\lambda-1) \sum_{k \geq 0} \lambda^{k}z^{\frac{-\nu\ellmahl^{k}}{\ellmahl-1}}.
\end{multline*}
The $z$-adic valuations of the last two terms of the right hand-side of the latter equality are positive. Applying Lemma \ref{lem:sol ordre un inhom mu z c un} again, we find 
\begin{multline*}
 g_{1}(\lambda,z)=
 -1 + (\lambda-1) \sum_{j\geq 0}\sum_{k \geq 1}  \lambda^{j+k} z^{\ellmahl^{j}\frac{\nu(1-\ellmahl^{k})}{\ellmahl-1}}
+
(\lambda-1) \sum_{j\geq 0}\sum_{k \geq 0}    \lambda^{j+k} z^{\frac{-\nu \ellmahl^{j+k}}{\ellmahl-1}}
 \\ =-1 + (\lambda-1) \sum_{j\geq 0}\sum_{k \geq 1}  \lambda^{j+k} z^{\ellmahl^{j}\frac{\nu(1-\ellmahl^{k})}{\ellmahl-1}}
+
(\lambda-1) \sum_{l \geq 0}  (l+1) \lambda^{l} z^{\frac{-\nu \ellmahl^{l}}{\ellmahl-1}}.
\end{multline*}

We now solve the equation \eqref{eq:second eq}. This equation can be written as follows
$$
(\lambda \malop{\ellmahl}-z^{\nu}) h(z)^{-1} (\lambda \malop{\ellmahl}-1)(g_{2}(\lambda,z))=z^{\nu} \theta_{\nu}(z) (\lambda-1)^{2}.
$$
Multiplying by $\theta_{\nu}(z)^{-1}z^{-\nu}$, we see that the latter equation is equivalent to
$$
(\lambda \malop{\ellmahl}-1) \theta_{\nu}(z)^{-1} h(z)^{-1}(\lambda \malop{\ellmahl}-1)(g_{2}(\lambda,z))=(\lambda-1)^{2}.
$$
Applying Lemma \ref{lem:sol ordre un inhom mu z c un}, we find that the latter equation is equivalent to
$$  
\theta_{\nu}(z)^{-1} h(z)^{-1}(\lambda \malop{\ellmahl}-1)(g_{2}(\lambda,z))=\lambda-1, 
$$
{\it i.e.}, to 
$$ 
(\lambda \malop{\ellmahl}-1)(g_{2}(\lambda,z))=(\lambda-1)\theta_{\nu}(z) h(z)=(\lambda-1)z^{\frac{\nu}{\ellmahl-1}}+(\lambda-1). 
$$
Applying Lemma \ref{lem:sol ordre un inhom mu z c un} again, we find 
$$
g_{2}(\lambda,z)=(\lambda-1)\sum_{k \leq -1} \lambda^{k}\malop{\ellmahl}^{k}(z^{\frac{\nu}{\ellmahl-1}}) + 1=(\lambda-1)\sum_{k \leq -1} \lambda^{k}z^{\frac{\ellmahl^{k}\nu}{\ellmahl-1}}+1. 
$$

Theorem \ref{theo:theo pas intro} guarantees that 
$$
y_{1}=\ev{1} (g_{1}(\lambda,z)e_{\lambda})=-e_{1}
$$ 
and 
$$
y_{2}=\ev{1}(\partial_{\lambda} (g_{2}(\lambda,z)e_{\lambda}))=(\sum_{k \leq -1} z^{\frac{\ellmahl^{k}\nu}{\ellmahl-1}})e_{1}+\logm_{1,1}
$$ 
are $\C$-linearly independent solutions of $L$.

\section{Factorization of Mahler operators}\label{sec: fact}

The aim of this section is to prove the following result relative to a Mahler operator 
\begin{equation}\label{}
L= a_{n}(z) \malop{\ellmahl}^{n} + a_{n-1}(z)\malop{\ellmahl}^{n-1} + \cdots + a_{0}(z) \in \mathcal{D}_{\Hahn}. 
\end{equation}
We let  
\begin{itemize}
 \item $
\mu_{1} < \cdots < \mu_{k}
$  be 
the slopes of $L$ with  respective multiplicities $r_{1},\ldots,r_{k}$; 
 \item $c_{i,1},\ldots,c_{i,r_{i}}$ be the exponents (repeated with multiplicities) of $L$ attached to the slope $\mu_{i}$. 
\end{itemize}

\begin{prop}\label{prop fact op}
The operator $L$ has a factorization 
$$
L=a(z)L_{k}\cdots L_{1} 
$$
where 
\begin{itemize}
 \item $a(z) \in \Hahn^{\times}$ is such that $\val(a(z))=\val(a_{0}(z))$;  
 \item $\cld a(z) = \prod_{i=1}^{k} \prod_{j=1}^{r_{i}} (-c_{i,j})^{-1} \cld a_{0}(z)$;
\item the $L_{i} \in \dqhahn$ are given by 
 $$
L_{i}= (z^{\nu_{i}}\malop{\ellmahl}- c_{i,r_{i}})h_{i,r_{i}}(z)^{-1}\cdots (z^{\nu_{i}}\malop{\ellmahl}-c_{i,1})h_{i,1}(z)^{-1}
$$
for some $h_{i,j}(z) \in \Hahn^{\times}$ tangent to the identity and with 
\begin{equation}\label{def nui}
 \nu_{i} =(\ellmahl-1) (\ellmahl^{r_{1}+\cdots+r_{i-1}}(\mu_{i}-\mu_{i-1})+\cdots+\ellmahl^{r_{1}}(\mu_{2}-\mu_{1})+\mu_{1}). 
\end{equation} 
\end{itemize}
\end{prop}

\begin{remark}
 This is a refinement of \cite[Theorem 15]{RoquesLSMS}, where a similar statement is proved but without the explicit values of the $\nu_{i}$ in terms of the $\mu_{j}$. 
\end{remark}

The proof of Proposition \ref{prop fact op} given in Section \ref{sec: proof prop fact op} will use some preliminary results gathered in the following two sections. 

\subsection{Preliminary results : basic properties of slopes and exponents}

In this Section, we collect basic results relative to the notions of slopes and exponents introduced in Sections \ref{sec : def slopes} and \ref{subsec:exponents}. 

\subsubsection{Slopes and exponents of the gauge transform $L^{[\theta_{\mu}(z)]}$}

We remind the notation
\begin{multline*}
 L^{[\theta_{\mu}(z)]}=\theta_{\mu}(z)^{-1}L\theta_{\mu}(z)\\ 
 =\sum_{i=0}^{n} z^{(1+\ellmahl+\cdots+\ellmahl^{i-1})\mu} a_i(z) \malop{\ellmahl}^i
=
\sum_{i=0}^{n} z^{\frac{\ellmahl^{i}-1}{\ellmahl-1}\mu} a_i(z) \malop{\ellmahl}^i
\end{multline*}
introduced in Section \ref{sec : def slopes} where 
$$
\theta_{\mu}(z)=z^{\frac{\mu}{\ellmahl-1}}
$$ 
with $\mu \in \Q$. We have $L^{[\theta_{\mu}(z)]}(f(z))=0$ if and only if $L(f(z)\theta_{\mu}(z))=0$.

%We have already seen that Moreover, we have :

\begin{lem}\label{lem:newton and gauge}
The slopes of $L^{[\theta_{\mu}(z)]}$ are 
$$
\mu_{1} + \frac{\mu}{\ellmahl-1} < \cdots < \mu_{k}+\frac{\mu}{\ellmahl-1}
$$ 
with respective multiplicities $r_{1},\ldots,r_{k}$. 
Moreover, 
$$ \chi(\mu_{j}+\frac{\mu}{\ellmahl-1},L^{[\theta_{\mu}(z)]};X)= \chi(\mu_{j},L;X)$$ 
and, hence, the exponents counted with multiplicities of $L^{[\theta_{\mu}(z)]}$ attached to the slope $\mu_{j}+\frac{\mu}{\ellmahl-1}$ coincide with those of $L$ attached to the slope $\mu_{j}$. 
\end{lem}

\begin{proof}
For the first assertions concerning the slopes and their multiplicities, see Proposition \ref{prop: slopes l theta mu}. 
Moreover, the equality 
$$ \chi(\mu_{j}+\frac{\mu}{\ellmahl-1},L^{[\theta_{\mu}(z)]};X)= \chi(\mu_{j},L;X)$$ 
follows immediately from the definition of the characteristic polynomials by using the equality  
$$(L^{[\theta_{\mu}(z)]})^{[\theta_{-(\ellmahl-1)(\mu_{j}+\frac{\mu}{\ellmahl-1})}(z)]}=L^{[\theta_{-(\ellmahl-1)\mu_{j}}(z)]}.$$ 
The very last assertion of the lemma follows from this equality of characteristic polynomials and from the definition of the exponents and of their multiplicities. 
\end{proof}

 \subsubsection{The gauge transform $L^{[e_{c}]}$}
We will also use the notation
$$
L^{[e_{c}]}=e_{c}^{-1}Le_{c}=\sum_{i=0}^{n} c^{i} a_i(z) \malop{\ellmahl}^i
$$
where $c \in \C^{\times}$.   
We have $L^{[e_{c}]}(f(z))=0$ if and only if $L(f(z)e_{c})=0$.

\begin{lem}\label{lem:newton and gauge caract}
The operators $L^{[e_{c}]}$ and $L$ have the same slopes $\mu_{1}  < \cdots < \mu_{k}$ 
with the same multiplicities $r_{1},\ldots,r_{k}$. Moreover, 
$$ \chi(\mu_{j},L^{[e_{c}]};X)= \chi(\mu_{j},L;cX)$$ 
and, hence, the exponents of $L^{[e_{c}]}$ and of $L$ attached to a given slope $\mu_{j}$ are related as follows :   
\begin{multline*}
\text{list of exponents of $L^{[e_{c}]}$ counted with mult. attached to the slope $\mu_{j}$}\\
= c^{-1}\cdot(\text{list of exponents of $L$ counted with mult. attached to the slope $\mu_{j}$}). 
\end{multline*}
\end{lem}

\begin{proof}
Since $\val(c^{i}a_{i})=\val(a_{i})$, we obviously have 
\begin{equation}\label{newt inv ec}
 \mathcal N(L^{[e_{c}]})=\mathcal N(L)
\end{equation}
and, hence, $L^{[e_{c}]}$ and $L$ have the same slopes with the same multiplicities. Moreover, the equality 
$$
\chi(\mu_{j},L^{[e_{c}]};X) =
\chi(\mu_{j},L;cX)
$$
follows immediately from the definition of the characteristic polynomials by using the equality 
$$(L^{[e_{c}]})^{[\theta_{-(\ellmahl-1)\mu_{j}}(z)]}=(L^{[\theta_{-(\ellmahl-1)\mu_{j}}(z)]})^{[e_{c}]}.
$$ 
The very last assertion of the lemma follows from the above equality of characteristic polynomials and from the definition of the exponents and of their multiplicities. 
\end{proof}

 \subsubsection{The gauge transform $L^{[g(z)]}$}
We will also use the notation
$$
L^{[g(z)]}=g(z)^{-1}Lg(z)=\sum_{i=0}^{n} \frac{\malop{\ellmahl}^{i}(g(z))}{g(z)} a_i(z) \malop{\ellmahl}^i
$$
for $g(z) \in \Hahn^{\times}$.   
We have $L^{[g(z)]}(f(z))=0$ if and only if $L(g(z)f(z))=0$.

\begin{lem}\label{lem:newton and gauge g}
If $\val(g(z))=0$ then the operators $L^{[g(z)]}$ and $L$ have the same slopes $\mu_{1}  < \cdots < \mu_{k}$ 
with the same multiplicities $r_{1},\ldots,r_{k}$. Moreover, 
$$ \chi(\mu_{j},L^{[g(z)]};X)= \chi(\mu_{j},L;X)$$ 
and, hence, the exponents counted with multiplicities of $L^{[g(z)]}$ attached to the slope $\mu_{j}$ coincide with those of $L$ attached to the slope $\mu_{j}$. 
\end{lem}

\begin{proof}
Since $\val(\frac{\malop{\ellmahl}^{i}(g(z))}{g(z)} a_{i}(z))=\val(a_{i}(z))$, we obviously have 
\begin{equation}\label{newt inv ec}
 \mathcal N(L^{[g(z)]})=\mathcal N(L)
\end{equation}
and, hence, $L$ and $L^{[g(z)]}$ have the same slopes with the same multiplicities. Moreover, the equality 
$$ \chi(\mu_{j},L^{[g(z)]};X)= \chi(\mu_{j},L;X)$$ 
follows immediately from the definition of the characteristic polynomials by using the equality  
$$(L^{[g(z)]})^{[\theta_{-(\ellmahl-1)\mu_{j}}(z)]}=(L^{[\theta_{-(\ellmahl-1)\mu_{j}}(z)]})^{[g(z)]},
$$ 
and the fact that the $\frac{\malop{\ellmahl}^{i}(g(z))}{g(z)} \in \Hahn^{\times}$ are tangent to the identity. 
The very last assertion of the lemma follows from this equality of characteristic polynomials and from the definition of the exponents and of their multiplicities. 
\end{proof}

\subsection{Preliminary results : on operators with smallest slope $0$}\label{sec:fact}

\begin{lem}\label{lem:smallest slope zero sol}
 Assume that $0$ is the smallest slope of $L$ ({\it i.e.}, $\mu_{1}=0$) and let $c$ be an exponent attached to the slope $0$. Then, there exist $L' \in  \dqhahn$ and $h(z)\in\Hahn^{\times}$ tangent to the identity such that $$
 L=L'(\malop{\ellmahl}-c)h(z)^{-1}.
 $$ 
\end{lem}

\begin{proof}
Indeed, \cite[Lemma 13]{RoquesLSMS} ensures that there exists $h(z) \in \Hahn^{\times}$ tangent to the identity such that 
$$
L(h(z)e_{c})=0.
$$ 
Since $h(z)e_{c}$ is also a solution of $(\malop{\ellmahl}-c)h(z)^{-1}$, we get, by right euclidean division, 
$$
L=L'(\malop{\ellmahl}-c)h(z)^{-1}
$$ 
for some $L' \in \dqhahn$. 
\end{proof}

\begin{lem}\label{lem:slopes L zero}
Assume that $\mathcal{S}(L) \subset \R^{+}$. Then, for any $c \in \C^{\times}$ and any $h(z) \in \Hahn^{\times}$ tangent to the identity, we have 
\begin{equation}\label{eq newt L phi moins c}
 \New(L(\malop{\ellmahl}-c)h(z)^{-1}) = \left([1,\ellmahl] \times [\val(a_{0}(z)),+\infty[ \right)\bigcup \varphi(\New(L)),  
\end{equation}
where $\varphi : \R^{2} \rightarrow \R^{2}$ is the map given by $\varphi(x,y)=(\ellmahl x,y)$. It follows that   
$$
\mathcal{S}(L(\malop{\ellmahl}-c)h(z)^{-1})=\ellmahl^{-1} \mathcal{S}(L) \cup \{0\}
$$ 
and, hence, for any $\lambda \in \mathcal{S}(L(\malop{\ellmahl}-c)h(z)^{-1})$, 
\begin{equation}\label{form multi}
r(\lambda,L(\malop{\ellmahl}-c)h(z)^{-1}) = 
\begin{cases}
 r(\ellmahl \lambda,L) \text{ if } \lambda \neq 0,\\
 r(0,L)+1  \text{ if } \lambda = 0.
\end{cases}
\end{equation}  
Moreover, for any $\lambda \in \mathcal{S}(L(\malop{\ellmahl}-c)h(z)^{-1})$, we have 
\begin{equation}\label{form poly char}
\chi(\lambda,L(\malop{\ellmahl}-c)h(z)^{-1};X)
%=\chi(\ellmahl \mu,L)\chi(\ellmahl \mu,\malop{\ellmahl}-c)
=
\begin{cases}
 \chi(\ellmahl \lambda,L;X) \text{ if } \lambda \neq 0,\\
 \chi(0,L;X)(X-c)  \text{ if } \lambda = 0.
\end{cases}
\end{equation}
\end{lem}

\begin{proof}
We will only treat the case when 
\begin{itemize}
 \item $0$ is a slope of $L$
 \item and $L$ has at least one nonzero slope, 
\end{itemize}
the other cases being similar. 

Without loss of generality, one can assume that $a_{0}(z)=1$.  
Moreover, combining the equality 
$$((L(\malop{\ellmahl}-c)h(z)^{-1})^{[e_{c}]})^{[h(z)]}=h(z)^{-1}cL^{[e_{c}]}(\malop{\ellmahl}-1)$$ 
with Lemma \ref{lem:newton and gauge caract} and Lemma \ref{lem:newton and gauge g}, we see that it is sufficient to treat the case $h(z)=1$ and $c=1$. 

We have  
$$ 
L(\malop{\ellmahl}-1)  
= b_{n+1}(z) \malop{\ellmahl}^{n+1} + b_{n}(z) \malop{\ellmahl}^{n}
+\cdots+ b_{1}(z) \malop{\ellmahl} +b_{0}(z) 
$$
with 
\begin{multline*}
 b_{0}(z)=-a_{0}(z), b_{1}(z)=a_{0}(z)-a_{1}(z), \ldots \\
 \ldots, b_{n}(z)=a_{n-1}(z)-a_{n}(z), b_{n+1}(z)=a_{n}(z).
\end{multline*}
In order to prove \eqref{eq newt L phi moins c}, we have to show : 
\begin{enumerate}[(i)]
 \item  $\val(b_{0}(z))=\val(a_{0}(z))$ and, for all $i \in \{1,\ldots,k\}$, $\val(b_{\alpha_{i}+1}(z))=\val(a_{\alpha_{i}}(z))$; 
 \item for $j \in \{0,\ldots,\alpha_{1}+1\}$, $\val(b_{j}(z)) \geq 0$;
\item for $i \in \{1,\ldots,k-1\}$, for $j \in \{\alpha_{i}+1,\ldots,\alpha_{i+1}\}$, we have 
\begin{equation}\label{eq ineg pentes}
 \frac{\val(b_{j}(z))-\val(b_{\alpha_{i}+1}(z))}{\ellmahl^{j}-\ellmahl^{\alpha_{i}+1}} \geq 
\frac{\val(b_{\alpha_{i+1}+1}(z))-\val(b_{\alpha_{i}+1}(z))}{\ellmahl^{\alpha_{i+1}+1}-\ellmahl^{\alpha_{i}+1}}. 
\end{equation}
\end{enumerate}

We have used the notations  
\begin{multline*}
 v_{0}=(\ellmahl^{\alpha_{0}},\val(a_{\alpha_{0}}(z)))=(p^{0},\val(a_{0}(z))),\ldots\\
\ldots,v_{i}=(\ellmahl^{\alpha_{i}},\val(a_{\alpha_{i}}(z))),\ldots\\
\ldots,v_{k}=(\ellmahl^{\alpha_{k}},\val(a_{\alpha_{k}}(z)))=(\ellmahl^{n},\val(a_{n}(z)))
\end{multline*}
introduced at the very beginning of Section \ref{sec : def slopes} for the vertices, ordered by increasing abscissa, of the polygon $\mathcal N(L)$.

We now proceed to the proof of the above properties (i), (ii) and (iii). 

It is clear that $\val(b_{0}(z))=0$ and $\val(b_{1}(z)) \geq 0,\ldots,\val(b_{\alpha_{1}}(z)) \geq 0$ because $b_{0}(z)=a_{0}(z)=1$ and because all the $a_{i}(z)$ have valuation $\geq 0$ (since $\mathcal{S}(L) \subset \R^{+}$).  
Also, for $i \in \{1,\ldots,k-1\}$, we have $\val(b_{\alpha_{i}+1}(z))=\val(a_{\alpha_{i}}(z))$ because $b_{\alpha_{i}+1}(z)=a_{\alpha_{i}}(z)-a_{\alpha_{i}+1}(z)$ and $\val(a_{\alpha_{i}+1}(z))>\val(a_{\alpha_{i}}(z))$ because $e_{i}=(\ellmahl^{\alpha_{i}},\val(a_{\alpha_{i}}(z)))$ is a vertex of $\mathcal{N}(L)$. The equality $\val(b_{\alpha_{k}+1}(z))=\val(a_{\alpha_{k}}(z))$ also holds true because $b_{\alpha_{k}+1}(z)=b_{n+1}(z)=a_{n}(z)=a_{\alpha_{k}}(z)$. This ensures that the above first two properties (i) and (ii) hold true. 

Moreover, for $i \in \{1,\ldots,k-1\}$ and $j \in \{\alpha_{i}+1,\ldots,\alpha_{i+1}\}$,  the inequality \eqref{eq ineg pentes} holds true 
because 
\begin{multline*}
 \frac{\val(b_{j}(z))-\val(b_{\alpha_{i}+1}(z))}{\ellmahl^{j}-\ellmahl^{\alpha_{i}+1}}
=
\frac{\val(b_{j}(z))-\val(a_{\alpha_{i}}(z))}{\ellmahl^{j}-\ellmahl^{\alpha_{i}+1}} \\
 \geq
\min \left\{\frac{\val(a_{j}(z))-\val(a_{\alpha_{i}}(z))}{\ellmahl^{j}-\ellmahl^{\alpha_{i}+1}},
\frac{\val(a_{j-1}(z))-\val(a_{\alpha_{i}}(z))}{\ellmahl^{j}-\ellmahl^{\alpha_{i}+1}}\right\}
\end{multline*}
(we have used the fact that 
$$\val(b_{j}(z)) \geq \min \{\val(a_{j-1}(z)),\val(a_{j}(z))\}).$$ 
But, on the one hand, we have  
\begin{multline*}
 \frac{\val(a_{j-1}(z))-\val(a_{\alpha_{i}}(z))}{\ellmahl^{j}-\ellmahl^{\alpha_{i}+1}} 
=
\frac{1}{\ellmahl} \frac{\val(a_{j-1}(z))-\val(a_{\alpha_{i}}(z))}{\ellmahl^{j-1}-\ellmahl^{\alpha_{i}}} \\ 
\geq 
\frac{1}{\ellmahl} \frac{\val(a_{\alpha_{i+1}}(z))-\val(a_{\alpha_{i}}(z))}{\ellmahl^{\alpha_{i+1}}-\ellmahl^{\alpha_{i}}}
=
\frac{\val(b_{\alpha_{i+1}+1}(z))-\val(b_{\alpha_{i}+1})(z)}{\ellmahl^{\alpha_{i+1}+1}-\ellmahl^{\alpha_{i}+1}},
\end{multline*}
the latter inequality being true because $e_{i}=(\ellmahl^{\alpha_{i}},\val(a_{\alpha_{i}}(z)))$ and $e_{i+1}=(\ellmahl^{\alpha_{i+1}},\val(a_{\alpha_{i+1}}(z)))$  are vertices of $\mathcal{N}(L)$. On the other hand, we have  
\begin{multline*}
 \frac{\val(a_{j}(z))-\val(a_{\alpha_{i}}(z))}{\ellmahl^{j}-\ellmahl^{\alpha_{i}+1}}
\geq 
\frac{\val(a_{j}(z))-\val(a_{\alpha_{i}}(z))}{\ellmahl^{j}-\ellmahl^{\alpha_{i}}}\\
\geq 
\frac{\val(a_{\alpha_{i+1}}(z))-\val(a_{\alpha_{i}}(z))}{\ellmahl^{\alpha_{i+1}}-\ellmahl^{\alpha_{i}}} 
= 
\ellmahl \frac{\val(b_{\alpha_{i+1}+1}(z))-\val(b_{\alpha_{i}+1}(z))}{\ellmahl^{\alpha_{i+1}+1}-\ellmahl^{\alpha_{i}+1}} \\ 
\geq 
\frac{\val(b_{\alpha_{i+1}+1}(z))-\val(b_{\alpha_{i}+1}(z))}{\ellmahl^{\alpha_{i+1}+1}-\ellmahl^{\alpha_{i}+1}}, 
\end{multline*}
the second inequality above being true because $e_{i}=(\ellmahl^{\alpha_{i}},\val(a_{\alpha_{i}}(z)))$ and $e_{i+1}=(\ellmahl^{\alpha_{i+1}},\val(a_{\alpha_{i+1}}(z)))$  are vertices of $\mathcal{N}(L)$. This concludes the proof of \eqref{eq ineg pentes} and, hence, of (iii) and of \eqref{eq newt L phi moins c}. 

It remains to justify the equalities \eqref{form multi} and  \eqref{form poly char}. Let 
$$
\lambda_{1}=\ellmahl^{-1} \mu_{1}=0 <\lambda_{2} = \ellmahl^{-1} \mu_{2} < \cdots < \lambda_{k} = \ellmahl^{-1} \mu_{k} 
$$
be the slopes of $L(\malop{\ellmahl}-1)$. 
Consider a nonzero slope $\lambda_{j}$ of $L(\malop{\ellmahl}-1)$.  So, $\ellmahl \lambda_{j}=\mu_{j}$ is a slope of $L$ and, setting 
$$
L^{[\theta_{-(\ellmahl-1) \ellmahl \lambda_{j}}(z)]} = \sum_{i=0}^{n} c_{i}(z) \malop{\ellmahl}^{i}, 
$$
we have $\val(c_{\alpha_{j}}(z))=\val(c_{\alpha_{j+1}}(z))$, $\val(c_{i}(z)) \geq \val(c_{\alpha_{j}}(z))$ for $i \in \{\alpha_{j},\ldots,\alpha_{j+1}\}$ and $\val(c_{i}) > \val(c_{\alpha_{j}}(z))$ for $i \in \{0,\ldots,n\} \setminus \{\alpha_{j},\ldots,\alpha_{j+1}\}$. The characteristic polynomial attached to the slope $\ellmahl \lambda_{j}=\mu_{j}$ of $L$ is given by 
$$
\chi(\ellmahl \lambda_{j},L;X) = \left(z^{-\val(L^{[\theta_{-(\ellmahl-1)\ellmahl\lambda_{j}}(z)]})}  \sum_{i=\alpha_{j-1}}^{\alpha_{j}} c_{i}(z)X^{i} \right)_{\vert z=0}. 
$$
On the other hand, we have 
$$
(L(\malop{\ellmahl}-1))^{[\theta_{-(\ellmahl-1)\lambda_{j}}(z)]}
=(L\malop{\ellmahl})^{[\theta_{-(\ellmahl-1)\lambda_{j}}(z)]}
-L^{[\theta_{-(\ellmahl-1)\lambda_{j}}(z)]}. 
$$
But, 
$$
(L\malop{\ellmahl})^{[\theta_{-(\ellmahl-1)\lambda_{j}}(z)]}
=
\theta_{(\ellmahl-1)(1-\ellmahl)\lambda_{j}}(z)L^{[\theta_{-(\ellmahl-1) \ellmahl \lambda_{j}}(z)]} \malop{\ellmahl}
$$
and 
$$
L^{[\theta_{-(\ellmahl-1)\lambda_{j}}(z)]}
=
(L^{[\theta_{-(\ellmahl-1) \ellmahl \lambda_{j}}(z)]})^{[\theta_{-(\ellmahl-1) (1-\ellmahl) \lambda_{j}}(z)]}. 
$$
It follows from this that 
$$
\val((L \malop{\ellmahl})^{[\theta_{-(\ellmahl-1)\lambda_{j}}(z)]}) < \val(L^{[\theta_{-(\ellmahl-1)\lambda_{j}}(z)]}). 
$$
So, 
\begin{multline*}
 \val((L(\malop{\ellmahl}-1))^{[\theta_{-(\ellmahl-1)\lambda_{j}}(z)]})\\
=
\val((L \malop{\ellmahl})^{[\theta_{-(\ellmahl-1)\lambda_{j}}(z)]})
=
\val(\theta_{(\ellmahl-1)(1-\ellmahl)\lambda_{j}}(z)L^{[\theta_{-(\ellmahl-1) \ellmahl \lambda_{j}}(z)]}
)
\end{multline*}
and, hence, 
\begin{multline*}
 \chi(\lambda_{j},L(\malop{\ellmahl}-1);X) = \\ 
 \left(z^{-\val(L^{[\theta_{-(\ellmahl-1)\ellmahl\lambda_{j}}(z)]})}\sum_{i=\alpha_{j-1}+1}^{\alpha_{j}+1} c_{i-1}(z) X^{i} \right)_{\vert z=0} 
= 
\chi(\ellmahl \lambda_{j},L;X). 
\end{multline*}
Last, the proof of the equality $\chi(0,L(\malop{\ellmahl}-1);X)=\chi(0,L;X)(X-1)$ is easy and left to the reader. 
\end{proof}

\subsection{Proof of Proposition \ref{prop fact op}}\label{sec: proof prop fact op}
We set 
$$
M=L^{[\theta_{-(\ellmahl-1)\mu_{1}}(z)]}. 
$$
Lemma \ref{lem:newton and gauge} ensures that 
\begin{itemize}
 \item $\mathcal{S}(M)=\{\mu_{1}-\mu_{1}=0,\mu_{2}-\mu_{1},\ldots,\mu_{k}-\mu_{1}\}$; 
 \item for $i \in \{1,\ldots,k\}$, $r(\mu_{i}-\mu_{1},M)=r_{i}$; 
 \item for $i \in \{1,\ldots,k\}$, the exponents of $M$ (counted with multiplicities) attached to the slope $\mu_{i}-\mu_{1}$ coincide with those of $L$ attached to the slope $\mu_{i}$, {\it i.e.}, 
$
 c_{i,1},\ldots,c_{i,r_{i}}. 
 $ 
\end{itemize}

We claim that $M$ has a decomposition of the form 
\begin{equation}\label{claim:decomp si slope zero}
 M=M_{r_{1}}(\malop{\ellmahl}-c_{1,r_{1}})h_{1,r_{1}}(z)^{-1}\cdots(\malop{\ellmahl}-c_{1,1})h_{1,1}(z)^{-1} 
\end{equation}
for some $h_{1,1}(z),\ldots,h_{1,r_{1}}(z) \in \Hahn^{\times}$ tangent to the identity and some $M_{r_{1}} \in \dqhahn$ such that  
\begin{itemize}
 \item $\mathcal{S}(M_{r_{1}})=\{\ellmahl^{r_{1}}(\mu_{2}-\mu_{1}),\ldots,\ellmahl^{r_{1}}(\mu_{k}-\mu_{1})\}$; 
 \item for $i \in \{2,\ldots,k\}$, $r(\ellmahl^{r_{1}}(\mu_{i}-\mu_{1}),M_{r_{1}})=r(\mu_{i},L)$; 
 \item for $i \in \{2,\ldots,k\}$, the exponents of $M_{r_{1}}$ (counted with  multiplicities) attached to the slope $\mu_{i}-\mu_{1}$ are 
$
 c_{i,1},\ldots,c_{i,r_{i}}
 $.
\end{itemize}
Indeed, Lemma \ref{lem:smallest slope zero sol} ensures that there exists $h_{1,1} \in \Hahn^{\times}$ tangent to the identity and $M_{1} \in \dqhahn$ such that 
$$
M=M_{1}(\malop{\ellmahl}-c_{1,1})h_{1,1}(z)^{-1}. 
$$ 

If $r_{1} = 1$ then Lemma \ref{lem:slopes L zero} ensures that
\begin{itemize}
\item $\mathcal{S}(M_{1})=\{\ellmahl(\mu_{2}-\mu_{1}),\ldots,\ellmahl(\mu_{k}-\mu_{1})\}$; 
\item for $i \in \{2,\ldots,k\}$, 
$
r(\ellmahl(\mu_{i}-\mu_{1}),M_{1})=r_{i}
$;  
\item for $i \in \{2,\ldots,k\}$, the exponents of $M_{1}$ (counted with multiplicities) attached to the slope $\ellmahl(\mu_{i}-\mu_{1})$ are 
$
 c_{i,1},\ldots,c_{i,r_{i}}
 $. 
\end{itemize}
This proves our claim when $r_{1}=1$. 

If $r_{1} \geq 2$ then Lemma \ref{lem:slopes L zero} ensures that 
\begin{itemize}
 \item $\mathcal{S}(M_{1})=\{0,\ellmahl(\mu_{2}-\mu_{1}),\ldots,\ellmahl(\mu_{k}-\mu_{1})\}
$; 
 \item $r(0,M_{1})=r_{1} -1$; 
 \item for $i \in \{2,\ldots,k\}$, 
$
r(\ellmahl(\mu_{i}-\mu_{1}),M_{1})=r_{i} 
$;  
\item the exponents of $M_{1}$ (counted with multiplicities) attached to the slope $0$ are 
$
c_{1,2},\ldots,c_{1,r_{1}}
$
whereas those associated to the slope $\ellmahl(\mu_{i}-\mu_{1})$, for $i \in \{2,\ldots,k\}$, are 
$
 c_{i,1},\ldots,c_{i,r_{i}}
 $.
\end{itemize}
Arguing as above, we get a decomposition of the form 
 $$
M=M_{2}(\malop{\ellmahl}-c_{1,2})h_{1,2}(z)^{-1}(\malop{\ellmahl}-c_{1,1})h_{1,1}(z)^{-1}
$$
for some $h_{1,2}(z) \in \Hahn^{\times}$ tangent to the identity and some $M_{2} \in \dqhahn$. 

If $r_{1}=2$, then Lemma \ref{lem:slopes L zero} ensures that 
\begin{itemize}
 \item $\mathcal{S}(M_{2})=\{\ellmahl^{2}(\mu_{2}-\mu_{1}),\ldots,\ellmahl^{2}(\mu_{k}-\mu_{1})\}$; 
 \item for $i \in \{2,\ldots,k\}$, 
$
r(\ellmahl^{2}(\mu_{i}-\mu_{1}),M_{2})=r_{i}
$; 
 \item for $i \in \{2,\ldots,k\}$, the exponents of $M_{2}$ (counted with  multiplicities) attached to the slope $\ellmahl^{2}(\mu_{i}-\mu_{1})$ are 
$
 c_{i,3},\ldots,c_{i,r_{i}}. 
 $
\end{itemize}
This proves our claim when $r_{1}=2$

In the general case (arbitrary $r_{1}$), our claim follows by an obvious iteration of the above arguments. \\

Using \eqref{claim:decomp si slope zero} and the identities $L=M^{[\theta_{(\ellmahl-1)\mu_{1}}(z)]}$ and $(\malop{\ellmahl}-c)^{[\theta_{(\ellmahl-1)\mu_{1}}(z)]}=z^{(\ellmahl-1)\mu_{1}}\malop{\ellmahl}-c$, we obtain a decomposition  of the form 
$$
L=N_{1} L_{1}
$$
where 
$$
L_{1}=(z^{(\ellmahl-1)\mu_{1}}\malop{\ellmahl}-c_{1,r_{1}})h_{1,r_{1}}(z)^{-1}\cdots(z^{(\ellmahl-1)\mu_{1}}\malop{\ellmahl}-c_{1,1})h_{1,1}(z)^{-1}
$$
for some $h_{1,1}(z),\ldots,h_{1,r_{1}}(z) \in \Hahn^{\times}$ tangent to the identity and for some $N_{1} \in \dqhahn$ such that  
\begin{itemize}
 \item $\mathcal{S}(N_{1}) = 
\{\ellmahl^{r_{1}}(\mu_{2}-\mu_{1})+\mu_{1},\ldots,\ellmahl^{r_{1}}(\mu_{k}-\mu_{1})+\mu_{1}\}$; 
 \item for $i \in \{2,\ldots,k\}$, $r(\ellmahl^{r_{1}}(\mu_{i}-\mu_{1})+\mu_{1},N_{1})=r_{i}$; 
 \item for $i \in \{2,\ldots,k\}$, the exponents of $N_{1}$ (counted with  multiplicities) attached to the slope $\ellmahl^{r_{1}}(\mu_{i}-\mu_{1})+\mu_{1}$ are 
$
 c_{i,1},\ldots,c_{i,r_{i}}
$.
\end{itemize}
These properties of $N_{1}=M_{r_{1}}^{[\theta_{(\ellmahl-1)\mu_{1}}(z)]}$ follow from the properties of $M_{r_{1}}$ listed above and from Lemma \ref{lem:newton and gauge}. \vskip 10pt

Applying what precedes to $N_{1}$ instead of $L$, we find a decomposition of the form  
$$
N_{1}=N_{2} L_{2}
$$ 
where 
\begin{multline*}
 L_{2}=(z^{(\ellmahl-1)(\ellmahl^{r_{1}}(\mu_{2}-\mu_{1})+\mu_{1})}\malop{\ellmahl}-c_{2,r_{2}})h_{2,r_{2}}(z)^{-1}\cdots \\
 \cdots  (z^{(\ellmahl-1)(\ellmahl^{r_{1}}(\mu_{2}-\mu_{1})+\mu_{1})}\malop{\ellmahl}-c_{2,1})h_{2,1}(z)^{-1}
\end{multline*}
for some $h_{2,1}(z),\ldots,h_{2,r_{2}}(z) \in \Hahn^{\times}$ tangent to the identity and 
for some $N_{2} \in \dqhahn$ such that  
\begin{itemize}
 \item $ \mathcal{S}(N_{2}) = 
\{\ellmahl^{r_{1}+r_{2}}(\mu_{3}-\mu_{2})+\ellmahl^{r_{1}}(\mu_{2}-\mu_{1})+\mu_{1},\ldots,\ellmahl^{r_{1}+r_{2}}(\mu_{k}-\mu_{2})+\ellmahl^{r_{1}}(\mu_{2}-\mu_{1})+\mu_{1}\}$; 
 \item for $i \in \{3,\ldots,k\}$,  $r(\ellmahl^{r_{1}+r_{2}}(\mu_{i}-\mu_{2})+\ellmahl^{r_{1}}(\mu_{2}-\mu_{1})+\mu_{1},N_{2})=r_{i}$; 
 \item for $i \in \{3,\ldots,k\}$,  the exponents of $N_{2}$ (and their multiplicities) attached to the slope $\ellmahl^{r_{1}+r_{2}}(\mu_{i}-\mu_{2})+\ellmahl^{r_{1}}(\mu_{2}-\mu_{1})+\mu_{1}$ are 
$
 c_{i,1},\ldots,c_{i,r_{i}} 
$. 
\end{itemize}

Iterating the previous construction, we find $L_{1},\ldots,L_{k} \in \dqhahn$ of the form announced in the statement of Proposition \ref{prop fact op} such that 
$$
L=a(z) L_{k} \cdots L_{1}
$$ 
for some $a(z) \in \Hahn^{\times}$. Equating the coefficients of degree $0$ in the latter equation, we find $a(z) \prod_{i=1}^{k} \prod_{j=1}^{r_{i}} (-c_{i,j} h_{i,j}(z)^{-1})=a_{0}(z)$, whence $\val(a(z))=\val(a_{0}(z))$ and 
$$
(\cld a(z)) \prod_{i=1}^{k} \prod_{j=1}^{r_{i}} (-c_{i,j}) = \cld a_{0}(z).
$$

\section{Forbenius Method : first justifications}\label{sec:frob method}

The aim of this section is to prove Proposition \ref{resol homogene pour frob} below, which is  
 relative to a Mahler operator 
\begin{equation}\label{op mahler pour preuves}
L= a_{n}(z) \malop{\ellmahl}^{n} + a_{n-1}(z)\malop{\ellmahl}^{n-1} + \cdots + a_{0}(z) \in \mathcal{D}_{\Hahn}. 
\end{equation}
In the rest of this section, we let  
$
\mu_{1} < \cdots < \mu_{k}
$  be 
the slopes of $L$ with  respective multiplicities $r_{1},\ldots,r_{k}$. Moreover, the multiplicity of an exponent $c$ of $L$ attached to the slope $\mu_{j}$ will be denoted by $m_{c,j} \in \Z_{\geq 1}$. If $c$ is a nonzero complex number that is not an exponent attached to the slope $\mu_{j}$, we set $m_{c,j}=0$.

\begin{prop}\label{resol homogene pour frob}
For any slope $\mu_{j}$, for any exponent $c$ of $L$ associated to the slope $\mu_{j}$, there exists a unique $g_{c,j}(\lambda,z)\in\Hahnlamb$ such that 
\begin{equation}\label{eq frob bis bis bis}
 L(g_{c,j}(\lambda,z)e_{\lambda}) = z^{\val(a_{0}(z))-\frac{\nu_{j}}{\ellmahl-1}}(\lambda-c)^{s_{c,j}+m_{c,j}} e_{\lambda}
\end{equation}
where 
$$
s_{c,j} = m_{c,1}+\cdots+m_{c,j-1}.
$$
and 
$$
\nu_{j} =(\ellmahl-1) (\ellmahl^{r_{1}+\cdots+r_{j-1}}(\mu_{j}-\mu_{j-1})+\cdots+\ellmahl^{r_{1}}(\mu_{2}-\mu_{1})+\mu_{1}). 
$$
The coefficients of $g_{c,j}(\lambda,z)$ have no pole at $\lambda=c$. Moreover, we have 
\begin{equation}\label{eq: form val g}
 \val(g_{c,j}(\lambda,z)) = -\mu_{j}
\end{equation}
and 
\begin{equation}\label{eq: form cld g}
\cld g_{c,j}(\lambda,z) = \lambda^{-r_{1}-\cdots-r_{j-1}} \frac{\prod_{i=1}^{j} \prod_{j=1}^{r_{i}} (-c_{i,j})}{\cld a_{0}}\frac{(\lambda-c)^{s_{c,j}+m_{c,j}}}{\prod_{l=1}^{r_{j}}(\lambda-c_{j,l})}.
\end{equation}
\end{prop} 

The proof of Proposition \ref{resol homogene pour frob}, given in Section \ref{sec: proof exist g}, uses some preliminary results gathered in the following section. 

\subsection{Preliminary results}\label{sec: prel res first just}

We say that a family $ (f_{i}(\lambda,z))_{i\in I}$ of elements of $\Hahnlamb$ is summable if the following properties are satisfied~: 
\begin{itemize}
 \item the set $\cup_{i \in I } \supp(f_{i}(\lambda,z))$ is well-ordered;
 \item for any $\gamma \in \Q$, the set 
$$\{i \in I \ \vert \  \gamma \in \supp(f_{i}(\lambda,z))\}$$ is finite.
\end{itemize}
 In this case, we define  
$$
\sum_{i\in I} f_{i}(\lambda,z) = \sum_{i \in I} (\sum_{\gamma \in \Q} f_{i,\gamma}(\lambda)) z^{\gamma} \in \Hahnlamb
$$
where $f_{i}(\lambda,z)=\sum_{\gamma \in \Q} f_{i,\gamma}(\lambda) z^{\gamma}$. 

In what follows, we let $\Hahnlamb^{<0}$ (resp. $\Hahnlamb^{>0}$) be the set made of the $f(\lambda,z) \in \Hahnlamb$ such that $\supp(f(\lambda,z)) \subset \Q_{<0}$ (resp. $\supp(f(\lambda,z)) \subset \Q_{>0}$).

\begin{lem}\label{lem: phik summable}
We have~:
\begin{itemize}
 \item if $g(\lambda,z) \in \Hahnlamb^{<0}$ then $(\malop{\ellmahl}^{k}(g(\lambda,z)))_{k \leq -1}$ is summable;
 \item  if $g(\lambda,z) \in \Hahnlamb^{>0}$ then $(\malop{\ellmahl}^{k}(g(\lambda,z)))_{k \geq 0}$ is summable. 
\end{itemize}
\end{lem}

\begin{proof}
 Let us first assume that $g(\lambda,z) \in \Hahnlamb^{<0}$. We set 
 $$
 E=\cup_{k \leq -1} E_{k}
 $$
 where
 $$
E_{k} = \supp(\malop{\ellmahl}^{k}(g(\lambda,z)))
 =
 \ellmahl^{k} \supp(g(\lambda,z)). 
 $$
 
 Let us first prove that $E$ is well-ordered. Let $F$ be a nonempty subset of $E$.  Since $E \subset \Q_{<0}$ and $\inf E_{k} \xrightarrow[k \rightarrow -\infty]{} 0$, we have 
 $$\inf F = \inf F \cap \cup_{k=k_{0}}^{-1} E_{k}
 $$ for some $k_{0} \leq -1$. Since $\supp(g(\lambda,z))$ is well-ordered, each $E_{k}=\ellmahl^{k} \supp(g(\lambda,z))$ is well-ordered and, hence, $\cup_{k=k_{0}}^{-1} E_{k}$ is well-ordered. It follows that $F \cap \cup_{k=k_{0}}^{-1} E_{k}$ and, hence, $F$ have a least element.  
 
In order to prove that $(\malop{\ellmahl}^{k}(g(\lambda,z)))_{k \leq -1}$ is summable, it remains to prove that, for any $\gamma \in \Q$, the set 
$\{k \leq -1 \ \vert \  \gamma \in E_{k}\}$ is finite. This is clear since $E_{k} \subset \Q_{<0}$ and $\inf E_{k} \xrightarrow[k \rightarrow -\infty]{} 0$.

The proof in the case $g(\lambda,z) \in \Hahnlamb^{>0}$ is similar.
\end{proof}

\begin{lem}\label{lem:sol ordre un inhom mu z c un}
For any $g(\lambda,z)=\sum_{\gamma \in \Q} g_{\gamma}(\lambda) z^{\gamma} \in \Hahnlamb$, there exists a unique $f(\lambda,z) \in \Hahnlamb$ such that 
\begin{equation}\label{eq inhom param mu z c un}
 (\lambda \malop{\ellmahl}-1)(f(\lambda,z))=g(\lambda,z).
\end{equation}
If 
$$g(\lambda,z)=g_{-}(\lambda,z)+g_{0}(\lambda)+g_{+}(\lambda,z)
$$ with $g_{-}(\lambda,z)
\in \Hahnlamb^{<0}$, $g_{0}(\lambda) \in \C$, $g_{+}(\lambda,z)
\in \Hahnlamb^{>0}$, then we have  
\begin{equation}\label{form for f un}
 f(\lambda,z)=  \sum_{k \leq -1} \lambda^{k}\malop{\ellmahl}^{k}(g_{-}(\lambda,z)) + \frac{g_{0}(\lambda)}{\lambda-1} -  \sum_{k \geq 0} \lambda^{k}\malop{\ellmahl}^{k}(g_{+}(\lambda,z)).
\end{equation}
\end{lem}
\begin{proof}
Lemma \ref{lem: phik summable} ensures that the families $(\lambda^{k}\malop{\ellmahl}^{k}(g_{-}(\lambda,z)))_{k \leq -1}$ and $(\lambda^{k}\malop{\ellmahl}^{k}(g_{+}(\lambda,z)))_{k \geq 0}$ are summable. The right-hand side of \eqref{form for f un} is thus meaningful and defines an element of $\Hahnlamb$.  
A straightforward calculation shows that this $f(\lambda,z)$ satisfies \eqref{eq inhom param mu z c un}. 

In order to prove that \eqref{form for f un} is the unique element of  $\Hahnlamb$ satisfying \eqref{eq inhom param mu z c un}, it is necessary and sufficient to prove that there is no  nonzero $h(\lambda,z) \in \Hahnlamb$ such that 
\begin{equation}\label{eq hom param hat}
 (\lambda \malop{\ellmahl}-1)(h(\lambda,z))=0. 
\end{equation}
If $h(\lambda,z) \in \Hahnlamb$ satisfies \eqref{eq hom param hat} then we have $\lambda \malop{\ellmahl}(h(\lambda,z))=h(\lambda,z)$. It follows that $\ellmahl \supp(h(\lambda,z))=\supp(\malop{\ellmahl}(h(\lambda,z)))=\supp(h(\lambda,z))$ and, hence, for all $k\in\Z$, $\ellmahl^{k} \supp(h(\lambda,z))=\supp(h(\lambda,z))$. Since $\supp(h(\lambda,z))$ is well-ordered, the latter equalities imply $\supp(h(\lambda,z)) \subset \{0\}$; indeed, if there were $\gamma \in  \supp(h(\lambda,z)) \cap \Q_{>0}$ (resp. $\gamma \in  \supp(h(\lambda,z)) \cap \Q_{<0}$) then $\ellmahl^{\Z_{\leq 0}} \gamma$ (resp. $\ellmahl^{\Z_{\geq 0}} \gamma$) would be a subset of $\supp(f(\lambda,z))$ with no least element, contradicting the fact that $\supp(f(\lambda,z))$ is well-ordered. So $h(\lambda,z)=h_{0}(\lambda) \in \C(\lambda)$. Inserting $h(\lambda,z)=h_{0}(\lambda)$ in  \eqref{eq hom param hat}, we find $h(\lambda,z)=h_{0}(\lambda)=0$, as expected. 
\end{proof}

\begin{lem}\label{lem:sol ordre un inhom}
Consider $\mu \in \Q$ and $c \in \C^{\times}$. For any $g(\lambda,z)=\sum_{\gamma \in \Q} g_{\gamma} (\lambda)z^{\gamma} \in \Hahnlamb$, there exists a unique $f(\lambda,z)=\sum_{\gamma \in \Q} f_{\gamma} (\lambda)z^{\gamma} \in \Hahnlamb$ such that 
\begin{equation}\label{eq inhom param}
 (z^{-\mu}\lambda \malop{\ellmahl}-c)(f(\lambda,z))=g(\lambda,z).
\end{equation}
Moreover,
\begin{itemize}
 \item if $\val(\theta_{-\mu}(z)g(\lambda,z)) \geq 0$ then $\val(f(\lambda,z))=\val(g(\lambda,z))$;
 \item if $\val(\theta_{-\mu}(z)g(\lambda,z)) < 0$ then $\val(\theta_{-\mu}(z)f(\lambda,z))=\frac{\val(\theta_{-\mu}(z)g(\lambda,z))}{\ellmahl}$;
\end{itemize}
and 
\begin{itemize}
 \item if $\val(\theta_{-\mu}(z)g(\lambda,z))<0$, then : 
 \begin{itemize}
\item $\cld(f(\lambda,z))=\lambda^{-1} \cld(g(\lambda,z))$; 
\item if the $g_{\gamma}(\lambda)$ have at most poles of order $\rho$ at $c$, then the $f_{\gamma}(\lambda)$  have at most poles of order $\rho+1$ at $c$;
\end{itemize}
 \item if $\val(\theta_{-\mu}(z)g(\lambda,z))=0$, then :
  \begin{itemize}
\item $\cld(f(\lambda,z))=(\lambda-c)^{-1}\cld(g(\lambda,z))$; 
\item  if the $g_{\gamma}(\lambda)$ have no pole at $c$ and if $\cld(g(\lambda,z))$ vanishes at $c$, then the $f_{\gamma}(\lambda)$  have no pole at $c$; 
\end{itemize}
\item if $\val(\theta_{-\mu}(z)g(\lambda,z))>0$, then :
  \begin{itemize}
\item $\cld(f(\lambda,z))=(-c)^{-1}\cld(g(\lambda,z))$; 
\item  if the $g_{\gamma}(\lambda)$ have no pole at $c$, then the $f_{\gamma}(\lambda)$  have no pole at $c$.
\end{itemize}
\end{itemize}
Last, if the $g_{\gamma}(\lambda)$ have at most poles of order $\rho$ at $c' \in \C^{\times} \setminus \{c\}$, then the $f_{\gamma}(\lambda)$  have at most poles of order $\rho$ at $c'$.
\end{lem} 

\begin{proof}
Replacing $\lambda$ by $c\lambda$, we can rewrite equation \eqref{eq inhom param} as follows 
 \begin{equation*}\label{eq inhom param hat}
 (\lambda \malop{\ellmahl}-1)(\theta_{-\mu}(z)f(c\lambda,z))=\theta_{-\mu}(z)c^{-1}g(c \lambda,z).
\end{equation*}
Combining this with Lemma \ref{lem:sol ordre un inhom mu z c un}, we get that there exists a unique $f(\lambda,z) \in \Hahnlamb$ satisfying \eqref{eq inhom param} and that it is given by 
\begin{multline}\label{sol c mu gen}
 \theta_{-\mu}f(\lambda,z) =  \sum_{k \leq -1} (c^{-1}\lambda)^{k}\malop{\ellmahl}^{k}(\widetilde{g}_{-}(c^{-1}\lambda,z))  \\ 
 + \frac{\widetilde{g}_{0}(c^{-1}\lambda)}{c^{-1}\lambda-1} 
 -  \sum_{k \geq 0} (c^{-1}\lambda)^{k}\malop{\ellmahl}^{k}(\widetilde{g}_{+}(c^{-1}\lambda,z))
\end{multline}
where $\widetilde{g}_{-}(\lambda,z)
\in \Hahnlamb^{<0}$, $\widetilde{g}_{0}(\lambda) \in \C$, $\widetilde{g}_{+}(\lambda,z),%=\sum_{\gamma \in \Q_{>0}} g_{\gamma}(\lambda) z^{\gamma} 
\in \Hahnlamb^{>0}$
are such that 
$$\theta_{-\mu}(z)c^{-1}g(c \lambda,z)=\widetilde{g}_{-}(\lambda,z)+\widetilde{g}_{0}(\lambda)+\widetilde{g}_{+}(\lambda,z), 
$$
{\it i.e.}, 
$$\widetilde{g}_{-}(\lambda,z)
=\theta_{-\mu}(z)\sum_{\gamma \in \Q_{<\frac{\mu}{\ellmahl-1}}} c^{-1}g_{\gamma}(c\lambda) z^{\gamma} 
\in \Hahnlamb^{<0},$$ 
$$\widetilde{g}_{0}(\lambda)=c^{-1}g_{\frac{\mu}{\ellmahl-1}}(c\lambda) \in \C(\lambda)$$
and  
$$\widetilde{g}_{+}(\lambda,z)=\theta_{-\mu}(z)\sum_{\gamma \in \Q_{>\frac{\mu}{\ellmahl-1}}} c^{-1}g_{\gamma}(c\lambda) z^{\gamma} 
\in \Hahnlamb^{>0}. 
$$

The properties of $f(\lambda,z)$ listed in the lemma follow by direct inspection of \eqref{sol c mu gen}. 
\end{proof}

\begin{lem}\label{lem:sol ordre un inhom bis}
Consider $\mu \in \Q$, $c_{1},\ldots,c_{r}  \in \C^{\times}$ and $h_{1}(z),\ldots,h_{r}(z) \in \Hahn^{\times}$ tangent to the identity. For any $g(\lambda,z) \in \Hahnlamb$, there exists a unique $f(\lambda,z) \in \Hahnlamb$ such that 
\begin{equation}\label{eq inhom param bis}
 (z^{-\mu}\lambda \malop{\ellmahl}-c_{r})h_{r}(z)^{-1}\cdots
 (z^{-\mu}\lambda \malop{\ellmahl}-c_{1})h_{1}(z)^{-1}(f(\lambda,z))=g(\lambda,z).
\end{equation}
Moreover,
\begin{itemize}
 \item if $\val(\theta_{-\mu}(z)g(\lambda,z)) \geq 0$ then $\val(f(\lambda,z))=\val(g(\lambda,z))$;
 \item if $\val(\theta_{-\mu}(z)g(\lambda,z)) < 0$ then $\val(\theta_{-\mu}(z)f(\lambda,z))=\frac{\val(\theta_{-\mu}(z)g(\lambda,z))}{\ellmahl^{r}}$;
\end{itemize}
and 
\begin{itemize}
 \item if $\val(\theta_{-\mu}(z)g(\lambda,z))<0$, then : 
 \begin{itemize}
\item $\cld(f(\lambda,z))=\lambda^{-r} \cld(g(\lambda,z))$; 
\item if the $g_{\gamma}(\lambda)$ have at most poles of order $\rho$ at $c \in \C^{\times}$, then the $f_{\gamma}(\lambda)$  have at most poles of order $\rho+m$ at $c$ where $m = \sharp \{i \in \{1,\ldots,r\} \ \vert \ c_{i}=c\}$;
\end{itemize}
 \item if $\val(\theta_{-\mu}(z)g(\lambda,z))=0$, then :
  \begin{itemize}
\item $\cld(f(\lambda,z))=(\lambda-c_{r})^{-1} \cdots (\lambda-c_{1})^{-1}\cld(g(\lambda,z))$; 
\item  if the $g_{\gamma}(\lambda)$ and $\cld(f(\lambda,z))$ have no pole at $c \in \C^{\times}$, then the $f_{\gamma}(\lambda)$  have no pole at $c$;
\end{itemize}
\item if $\val(\theta_{-\mu}(z)g(\lambda,z))>0$, then :
  \begin{itemize}
\item $\cld(f(\lambda,z))=(-c_{r})^{-1}\cdots (-c_{1})^{-1}\cld(g(\lambda,z))$; 
\item  if the $g_{\gamma}(\lambda)$ have no pole at $c \in \C^{\times}$, then the $f_{\gamma}(\lambda)$  have no pole at $c$. 
\end{itemize}
\end{itemize}
Last, if the $g_{\gamma}(\lambda)$ have at most poles of order $\rho$ at $c' \in \C^{\times} \setminus \{c_{1},\ldots,c_{r}\}$, then the $f_{\gamma}(\lambda)$  have at most poles of order $\rho$ at $c'$.
\end{lem} 

\begin{proof}
The equation \eqref{eq inhom param bis} is equivalent to the system of equations 
$$
\begin{cases}
(z^{-\mu}\lambda \malop{\ellmahl}-c_{r})h_{r}(z)^{-1}(f_{r}(\lambda,z))=g(\lambda,z)\\
(z^{-\mu}\lambda \malop{\ellmahl}-c_{r-1})h_{r-1}(z)^{-1}(f_{r-1}(\lambda,z))=f_{r}(\lambda,z)\\
\cdots\\
(z^{-\mu}\lambda \malop{\ellmahl}-c_{1})h_{1}(z)^{-1}(f_{1}(\lambda,z))=f_{2}(\lambda,z)\\
f(\lambda,z)=f_{1}(\lambda,z).
\end{cases}
$$
The result follows by $r$ successive applications of Lemma \ref{lem:sol ordre un inhom}. 
\end{proof}

\subsection{Proof of Proposition \ref{resol homogene pour frob}}\label{sec: proof exist g}

Using the fact that $g(\lambda,z) \in \Hahn_{\C(\lambda)}$ satisfies \begin{equation}\label{eq frob bis bis bis bis}
 L(g(\lambda,z)e_{\lambda}) = z^{\val(a_{0}(z))-\frac{\nu_{j}}{\ellmahl-1}}(\lambda-c)^{s_{c,j}+m_{c,j}} e_{\lambda}
\end{equation}
 if and only if $f(\lambda,z) = (\lambda-c)^{-s_{c,j}}g(\lambda,z) \in \Hahn_{\C(\lambda)}$ satisfies 
\begin{eqnarray}\label{eq frob bis bis bis bis gau}
 L(f(\lambda,z)e_{\lambda}) &=& z^{\val(a_{0}(z))-\frac{\nu_{j}}{\ellmahl-1}} (\lambda-c)^{m_{c,j}}e_{\lambda}\\
 &=&z^{\val(a_{0}(z))}\theta_{-\nu_{j}}(z) (\lambda-c)^{m_{c,j}}e_{\lambda}, \nonumber
\end{eqnarray}
we see that, in order to prove Proposition \ref{resol homogene pour frob}, it is sufficient to prove that~: 
\begin{enumerate}[(i)]
 \item there exists a unique $f(\lambda,z)$ in $ \Hahnlamb$ satisfying \eqref{eq frob bis bis bis bis};
\end{enumerate}
and that this $f(\lambda,z)$ has the following properties~: 
\begin{enumerate}[(i)]
\setcounter{enumi}{1}
 \item $\val(f(\lambda,z))=-\mu_{j}$; 
 \item $\cld f(\lambda,z) = \lambda^{-r_{1}-\cdots-r_{j-1}} \frac{\prod_{i=1}^{j} \prod_{j=1}^{r_{i}} (-c_{i,j})}{\cld a_{0}(z)}\frac{(\lambda-c)^{m_{c,j}}}{\prod_{l=1}^{r_{j}}(\lambda-c_{j,l})}$;
 \item the coefficients of 
$f(\lambda,z)$ have poles of order at most $s_{c,j}$ at $c$. 
\end{enumerate}

In order to prove these claims, let us first note that  \eqref{eq frob bis bis bis bis gau} can be rewritten as 
\begin{equation}\label{eq frob bis}
M(\theta_{\nu_{j}}(z) f(\lambda,z))=z^{\val(a_{0}(z))} (\lambda-c)^{m_{c,j}}
\end{equation}
where 
$$
M=(L^{[\theta_{-\nu_{j}}(z)]})^{[e_{\lambda}]}.  
$$
But, from the factorization  
$$
L=a(z)L_{k}\cdots L_{1}
$$ 
given by Proposition \ref{prop fact op}, we get the factorization  
$$
M=a(z)M_{k}\cdots M_{1}
$$ 
where $a(z) \in \Hahn^{\times}$ is such that $\val(a(z))=\val(a_{0}(z))$ and where 
\begin{multline*}
 M_{i}
=
(L_{i}^{[\theta_{-\nu_{j}}(z)]})^{[e_{\lambda}]}\\ 
= 
(z^{\nu_{i}-\nu_{j}}\lambda \malop{\ellmahl}- c_{i,r_{i}})h_{i,r_{i}}(z)^{-1}\cdots (z^{\nu_{i}-\nu_{j}} \lambda \malop{\ellmahl}-c_{i,1})h_{i,1}(z)^{-1}.
\end{multline*}
So, the equation \eqref{eq frob bis} can be rewritten as follows :  
\begin{equation}\label{eq:syst fact}
 \begin{cases}
M_{k}(f_{k}(\lambda,z)) = \frac{z^{\val(a_{0}(z))}}{a(z)}(\lambda-c)^{m_{c,j}} \\
M_{k-1}(f_{k-1}(\lambda,z)) = f_{k}(\lambda,z) \\
\cdots \\
M_{1}(f_{1}(\lambda,z)) = f_{2}(\lambda,z)\\
\theta_{\nu_{j}}(z) f(\lambda,z)=f_{1}(\lambda,z). 
\end{cases}
\end{equation}
Now, Lemma \ref{lem:sol ordre un inhom bis} ensures that there exists a unique $k$-uple 
$$(f_{1}(\lambda,z),\ldots,f_{k}(\lambda,z)) \in \Hahnlamb^{k}
$$ 
satisfying the first $k$ equations of \eqref{eq:syst fact}, whence the existence and the uniqueness of $f(\lambda,z) \in \Hahn$ satisfying \eqref{eq frob bis bis bis bis}; it is given by $ f(\lambda,z)=\theta_{-\nu_{j}}(z)f_{1}(\lambda,z)$. This proves claim (i).

Note that $\val(\frac{z^{\val(a_{0}(z))}}{a(z)})=0$ and set 
$$
\alpha = \cld \frac{z^{\val(a_{0}(z))}}{a(z)} =\cld a(z)^{-1}. 
$$
For $i \in \{j+1,\ldots,k\}$, we have $\nu_{i}-\nu_{j}>0$, so Lemma \ref{lem:sol ordre un inhom bis} ensures that
$f_{k}(\lambda,z),\ldots,f_{j+1}(\lambda,z)$ have $z$-adic valuation $0$ with constant terms 
\begin{eqnarray*}
 \cld{f_{k}(\lambda,z)}&=&\left(\prod_{l=1}^{r_{k}}(-c_{k,l})^{-1}\right) \alpha (\lambda-c)^{m_{c,j}}, \\ 
 \cld{f_{k-1}(\lambda,z)}&=&\left(\prod_{i=k-1}^{k}\prod_{l=1}^{r_{i}}(-c_{i,l})^{-1}\right)  \alpha(\lambda-c)^{m_{c,j}},\\ 
 \ldots & \ldots & \ldots \\
  \cld{f_{j+1}(\lambda,z)}&=&\left(\prod_{i=j+1}^{k}\prod_{l=1}^{r_{i}}(-c_{i,l})^{-1}\right) \alpha(\lambda-c)^{m_{c,j}}
\end{eqnarray*}
and also that the coefficients of $f_{k}(\lambda,z),\ldots,f_{j+1}(\lambda,z)$ have no pole at $\lambda=c$. 

For $i=j$, we have $\nu_{i}-\nu_{j}=0$, so Lemma \ref{lem:sol ordre un inhom bis} ensures that $f_{j}(\lambda,z)$  has $z$-adic valuation $0$ with constant term 
$$
\cld{f_{j}(\lambda,z)}=\left(\prod_{i=j+1}^{k}\prod_{l=1}^{r_{i}}(-c_{i,l})^{-1}\right) \alpha \frac{(\lambda-c)^{m_{c,j}}}{\prod_{l=1}^{r_{j}}(\lambda-c_{j,l})}
$$
and also that the coefficients of $f_{j}(\lambda,z)$ have no pole at $\lambda=c$. 

We have $\val(\theta_{\nu_{j-1}-\nu_{j}}(z)f_{j}(\lambda,z))=\val(\theta_{\nu_{j-1}-\nu_{j}}(z)) <0$, so Lemma \ref{lem:sol ordre un inhom bis} ensures that  
 $$
 \val(\theta_{\nu_{j-1}-\nu_{j}}(z)h_{j-1}(\lambda,z))=\frac{\val(\theta_{\nu_{j-1}-\nu_{j}}(z)h_{j}(\lambda,z))}{\ellmahl^{r_{j-1}}}<0.
 $$ 
Therefore, we have 
\begin{multline*}
 \val(\theta_{\nu_{j-2}-\nu_{j}}(z)f_{j-1}(\lambda,z))\\ 
=\val(\theta_{\nu_{j-2}-\nu_{j-1}}(z))
+\val(\theta_{\nu_{j-1}-\nu_{j}}(z)f_{j-1}(\lambda,z))
<0 
\end{multline*}
and Lemma \ref{lem:sol ordre un inhom bis} ensures that 
$$
\val(\theta_{\nu_{j-2}-\nu_{j}}(z)f_{j-2}(\lambda,z))=\frac{\val(\theta_{\nu_{j-2}-\nu_{j}}(z)h_{j-1}(\lambda,z))}{\ellmahl^{r_{j-2}}}<0.
$$
An obvious iteration of this argument leads to the fact that, for $i \in \{1,\ldots,j-1\}$, we have 
 $$
 \val(\theta_{\nu_{i}-\nu_{j}}(z)h_{i}(\lambda,z))=\frac{\val(\theta_{\nu_{i}-\nu_{j}}(z)f_{i+1}(\lambda,z))}{\ellmahl^{r_{i}}} <0. 
 $$

Therefore, we have 
\begin{multline*}
 \val(\theta_{\nu_{1}-\nu_{j}}(z)f_{1}(\lambda,z))=\frac{\val(\theta_{\nu_{1}-\nu_{j}}(z)h_{2}(\lambda,z))}{\ellmahl^{r_{1}}}\\
 =\frac{\val(\theta_{\nu_{1}-\nu_{2}}(z))}{\ellmahl^{r_{1}}}+\frac{\val(\theta_{\nu_{2}-\nu_{j}}(z)h_{2}(\lambda,z))}{\ellmahl^{r_{1}}}\\
 =\frac{\val(\theta_{\nu_{1}-\nu_{2}}(z))}{\ellmahl^{r_{1}}}+\frac{\val(\theta_{\nu_{2}-\nu_{j}}(z)h_{3}(\lambda,z))}{\ellmahl^{r_{1}+r_{2}}}\\
 =\cdots \\
 =\frac{\val(\theta_{\nu_{1}-\nu_{2}}(z))}{\ellmahl^{r_{1}}}
 +
 \frac{\val(\theta_{\nu_{2}-\nu_{3}}(z))}{\ellmahl^{r_{1}+r_{2}}}
 +\cdots + 
  \frac{\val(\theta_{\nu_{j-1}-\nu_{j}}(z))}{\ellmahl^{r_{1}+r_{2}+\cdots+r_{j-1}}}\\
  =  \frac{1}{\ellmahl-1}\left(\frac{{\nu_{1}-\nu_{2}}}{\ellmahl^{r_{1}}}
 +
 \frac{{\nu_{2}-\nu_{3}}}{\ellmahl^{r_{1}+r_{2}}}
 +\cdots + 
  \frac{{\nu_{j-1}-\nu_{j}}}{\ellmahl^{r_{1}+r_{2}+\cdots+r_{j-1}}}\right) \\
  =  \frac{\ellmahl^{r_{1}}(\mu_{1}-\mu_{2})}{\ellmahl^{r_{1}}}
 +
 \frac{\ellmahl^{r_{1}+r_{2}}(\mu_{2}-\mu_{3})}{\ellmahl^{r_{1}+r_{2}}}
 +\cdots + 
  \frac{\ellmahl^{r_{1}+r_{2}+\cdots+r_{j-1}}(\mu_{j-1}-\mu_{j})}{\ellmahl^{r_{1}+r_{2}+\cdots+r_{j-1}}}  \\
  = \mu_{1}-\mu_{j}
\end{multline*}
so 
$$
 \val(\theta_{-\nu_{j}}(z)f_{1}(\lambda,z)) = -\mu_{j}.
$$
This proves (ii).

It follows also from Lemma \ref{lem:sol ordre un inhom bis} that 
\begin{multline*}
 \cld(f_{1}(\lambda,z)) = \lambda^{-r_{1}-\cdots-r_{j-1}} \left(\prod_{i=j+1}^{k}\prod_{l=1}^{r_{i}}(-c_{i,l})^{-1} \right)\alpha\frac{(\lambda-c)^{m_{c,j}}}{\prod_{l=1}^{r_{j}}(\lambda-c_{j,l})}
\\
=\lambda^{-r_{1}-\cdots-r_{j-1}} \left(\prod_{i=j+1}^{k}\prod_{l=1}^{r_{i}}(-c_{i,l})^{-1} \right)\frac{\prod_{i=1}^{k} \prod_{j=1}^{r_{i}} (-c_{i,j})}{\cld a_{0}(z)}\frac{(\lambda-c)^{m_{c,j}}}{\prod_{l=1}^{r_{j}}(\lambda-c_{j,l})}
\\
=\lambda^{-r_{1}-\cdots-r_{j-1}} \frac{\prod_{i=1}^{j} \prod_{j=1}^{r_{i}} (-c_{i,j})}{\cld a_{0}(z)}\frac{(\lambda-c)^{m_{c,j}}}{\prod_{l=1}^{r_{j}}(\lambda-c_{j,l})}
\end{multline*}
(we have used the formula for $\cld a(z)$ given by Proposition \ref{prop fact op} for the second equality) and that the coefficients of $f_{j-1}(\lambda,z),\ldots,f_{1}(\lambda,z)$ have poles order at most $m_{c,j-1},m_{c,j-1}+m_{c,j-2},\ldots,m_{c,j-1}+m_{c,j-2}+\cdots+m_{c,1}$ at $c$ respectively. This prove (iii) and (iv).

\section{Frobenius method : last justifications}\label{sec: frob meth last just}

We use the notations introduced at the very beginning of Section \ref{sec:frob method}. We have seen in Proposition \ref{resol homogene pour frob} that, for any exponent $c$ associated to the slope $\mu_{j}$ of $L$, there exists a unique $g_{c,j}(\lambda,z)\in\Hahnlamb$ such that 
\begin{equation*}
 L(g_{c,j}(\lambda,z)e_{\lambda}) = z^{\val(a_{0}(z))-\frac{\nu_{j}}{\ellmahl-1}}(\lambda-c)^{s_{c,j}+m_{c,j}} e_{\lambda}
\end{equation*}
where 
$$
s_{c,j} = m_{c,1}+\cdots+m_{c,j-1}.
$$
and 
$$
\nu_{j} =(\ellmahl-1) (\ellmahl^{r_{1}+\cdots+r_{j-1}}(\mu_{j}-\mu_{j-1})+\cdots+\ellmahl^{r_{1}}(\mu_{2}-\mu_{1})+\mu_{1})
$$
and that the coefficients of $g_{c,j}(\lambda,z)$ have no pole at $\lambda=c$ and, hence,  $g_{c,j}(\lambda,z)e_{\lambda}$ belongs to 
$\Rspe_{\lambda,c}$. 

We first prove: 
\begin{prop}\label{prop: les ycjm sol}
 For any $m \in \{0,\ldots,m_{c,j}-1\}$,  
$$
y_{c,j,m}=\ev{c}(\partial_{\lambda}^{s_{c,j}+m} (g_{c,j}(\lambda,z)e_{\lambda}))
$$ 
is a solution of $L(y)=0$. 
\end{prop}
 
\begin{proof}
Using the fact that $\ev{c}$ and $\partial_{\lambda}$ are $\Hahn$-linear and commute with $\malop{\ellmahl}$, we see that, for any $m \in \{0,\ldots,m_{c,j}-1\}$, 
\begin{multline*}
 L(y_{c,j,m})\\ 
 =L(\ev{c}(\partial_{\lambda}^{s_{c,j}+m} (g_{c,j}(\lambda,z)e_{\lambda})))
 =\ev{c} \partial_{\lambda}^{s_{c,j}+m} (L (g_{c,j}(\lambda,z)e_{\lambda}))\\ 
 =z^{\val(a_{0}(z))-\frac{\nu_{j}}{\ellmahl-1}}\ev{c}\partial_{\lambda}^{s_{c,j}+m} ((\lambda-c)^{s_{c,j}+m_{c,j}} e_{\lambda})). 
\end{multline*}
Since $m \in \{0,\ldots,m_{c,j}-1\}$, we have $s_{c,j}+m<s_{c,j}+m_{c,j}$ and, hence, 
$$
\ev{c}\partial_{\lambda}^{s_{c,j}+m} ((\lambda-c)^{s_{c,j}+m_{c,j}} e_{\lambda}))=0. 
$$
This proves that $y_{c,j,m}$ is a solution of $L$ as claimed. 
\end{proof}

It remains to prove the following result. 

\begin{theo}\label{theo:theo pas du tout intro}
 We have attached to any slope $\mu_{j}$, to any exponent $c$ attached to the slope $\mu_{j}$ and to any $m \in \{0,\ldots,m_{c,j}-1\}$, a solution 
$$
y_{c,j,m}=\ev{c}(\partial_{\lambda}^{s_{c,j}+m} (g_{c,j}(\lambda,z)e_{\lambda}))
$$ 
of $L(y)=0$.   
These $n$ solutions are $\C$-linearly independent. 
\end{theo}

The proof is given in the following Section.   

\subsection{Proof of Theorem \ref{theo:theo pas du tout intro}}\label{subsec: preuve theo principal}

Using the Leibniz rule, we see that    
\begin{eqnarray*}
 y_{c,j,m} & \in & \operatorname{Span}_{\Hahn} (\ev{c}(\partial_{\lambda}^{0} (e_{\lambda})),\ev{c}(\partial_{\lambda}^{1} (e_{\lambda})),\ldots \\
&& \ \ \ \ \ \ \ \ \ \ \ \ \ \ \ \ \ \ \ \ \ \ \ \ \ \ \ \ \ \ \ \ \ \ \ \ \ \ \ \ \ldots ,\ev{c}(\partial_{\lambda}^{s_{c,j}+m_{c,j}-1} (e_{\lambda}))) \\ 
 &=& \operatorname{Span}_{\Hahn} (e_{c},\logm_{c,1},\ldots,\logm_{c,s_{c,j}+m_{c,j}-1}) \subset \operatorname{Span}_{\Hahn} (\logm_{c,j})_{j \geq 0}. 
\end{eqnarray*}
But, Lemma \ref{les logm sont indep} proven below guarantees that the family $(\logm_{c,j})_{c \in \C^{\times}, j \geq 0}$ is $\Hahn$-linearly independent. So, in order to prove Theorem \ref{theo:theo pas du tout intro}, it is sufficient to prove that, for any exponent $c$ of $L$, the family  $(y_{c,j,m})_{j \in \{1,\ldots,k\}, m \in \{0,\ldots,m_{c,j}-1\}}$ is $\C$-linearly independent. Let us prove this. Fix such a $c$ and consider a family $(a_{c,j,m})_{j \in \{1,\ldots,k\}, m \in \{0,\ldots,m_{c,j}-1\}}$ of complex numbers such that 
\begin{equation}\label{eq: lin comb ycjm}
 \sum_{j \in \{1,\ldots,k\}, m \in \{0,\ldots,m_{c,j}-1\}} a_{c,j,m} y_{c,j,m} =0. 
\end{equation}
We have to prove that the $a_{c,j,m}$ are all $0$. In this respect, we will use the following result. 

\begin{lem}
We have a decomposition of the form
\begin{multline}\label{eq:decomp ycjl}
 y_{c,j,m} = \sum_{u=0}^{s_{c,j}+m} h_{c,j,m,u}(z)\logm_{c,u} \\ 
 = \sum_{u=0}^{m-1} h_{c,j,m,u}(z) \logm_{c,u} + 
 h_{c,j,m,m} \logm_{c,m}
 +
 \sum_{u=m+1}^{s_{c,j}+m} h_{c,j,m,u}(z)\logm_{c,u} 
\end{multline}
for some $h_{c,j,m,u}(z) \in \Hahn$ such that 
\begin{equation}\label{eq:val des hcjlu}
\begin{cases}
 \val( h_{c,j,m,u}(z)) \geq -\mu_{j} \text{ for } u \in \{0,\ldots,m-1\},\\ 
  \val( h_{c,j,m,m}(z))=-\mu_{j},\\ 
 \val( h_{c,j,m,u}(z)) > -\mu_{j} \text{ for } u \in \{m+1\ldots,s_{c,j}+m\}. 
\end{cases}
\end{equation}
\end{lem}

\begin{proof}
 Using Leibniz rule, we obtain 
$$ 
{\partial_{\lambda}^{s_{c,j}+m} (g_{c,j}(\lambda,z)e_{\lambda})}
 =
\sum_{u=0}^{s_{c,j}+m} u!\binom{s_{c,j}+m}{u} {\partial_{\lambda}^{s_{c,j}+m-u} (g_{c,j}(\lambda,z)})\logm_{\lambda,u}
$$
and, hence, 
$$
y_{c,j,m} =\ev{c} {\partial_{\lambda}^{s_{c,j}+m} (g_{c,j}(\lambda,z)e_{\lambda})}
 = \sum_{u=0}^{s_{c,j}+m} h_{c,j,m,u} \logm_{c,u}.
$$
with 
$$
h_{c,j,m,u}(z)=u! \binom{s_{c,j}+m}{u} \ev{c}\partial_{\lambda}^{s_{c,j}+m-u} (g_{c,j}(\lambda,z)) \in \Hahn. 
$$

Accordingly to \eqref{eq: form val g}, we have  
$$
\val(g_{c,j}(\lambda,z)) = -\mu_{j}
$$
and, hence, 
$$
\val(h_{c,j,m,u}(z)) \geq \val(g_{c,j}(\lambda,z)) = -\mu_{j}. 
$$
Moreover, the latter inequality is an equality if and only if 
\begin{equation}\label{eq: equal evdelcld}
 \ev{c}\partial_{\lambda}^{s_{c,j}+m-u} (\cld g_{c,j}(\lambda,z)) \neq 0
\end{equation}
and it is a strict inequality if and only if 
\begin{equation}\label{eq: strict evdelcld}
 \ev{c}\partial_{\lambda}^{s_{c,j}+m-u} (\cld  g_{c,j}(\lambda,z)) = 0. 
\end{equation}
But, using \eqref{eq: form cld g}, we see that $\cld  g_{c,j}(\lambda,z)$ is a rational function in $\lambda$ with $(\lambda-c)$-adic valuation $s_{c,j}$. So, \eqref{eq: equal evdelcld} holds true if $s_{c,j}+m-u=s_{c,j}$ and that \eqref{eq: strict evdelcld} holds true if $s_{c,j}+m-u < s_{c,j}$. Whence the result. 
\end{proof}

Inserting \eqref{eq:decomp ycjl} in \eqref{eq: lin comb ycjm}, we get 
\begin{equation}\label{eq: lin comb hcjm}
   \sum_{j \in \{1,\ldots,k\}, m \in \{0,\ldots,m_{c,j}-1\}} \sum_{u=0}^{s_{c,j}+m} a_{c,j,m} h_{c,j,m,u}(z) \logm_{c,u}. 
  \end{equation}
Using the fact that the family $(\logm_{c,j})_{j \geq 0}$ is $\Hahn$-linearly independent (see Lemma \ref{les logm sont indep}  below), we get, for all $u \in \{0,\ldots,s_{c,k+1}-1\}$, 
 \begin{equation}\label{comb lin h}
 \sum_{\substack{j \in \{1,\ldots,k\}, m \in \{0,\ldots,m_{c,j}-1\} \\ \text{such that } u \in \{0,\ldots,s_{c,j}+m\}}} a_{c,j,m} h_{c,j,m,u}(z)=0. 
\end{equation}
But, using \eqref{eq:val des hcjlu}, we see that 
\begin{itemize}
 \item for $u=m_{c,k}-1$, all the terms in \eqref{comb lin h}, with the possible exception of  the term $a_{c,k,m_{c,k}-1}h_{c,k,m_{c,k}-1,m_{c,k}-1}(z)$ corresponding to $j=k$ and $m=m_{c,k}-1$, have $z$-adic valuation $>-\mu_{k}$; 
 \item  $\val(h_{c,k,m_{c,k}-1,m_{c,k}-1}(z))=-\mu_{k}$. 
 \end{itemize}
It follows that $a_{c,k,m_{k}-1}=0$.

Similarly, using \eqref{eq:val des hcjlu}, we see that  
\begin{itemize}
 \item for $u=m_{c,k}-2$, all the nonzero terms in \eqref{comb lin h}, with the possible exception of  the term $a_{c,k,m_{c,k}-2}h_{c,k,m_{c,k}-2,m_{c,k}-2}(z)$ corresponding to $j=k$ and $m=m_{c,k}-2$ have $z$-adic valuation $>-\mu_{k}$; 
 \item $\val(h_{c,k,m_{k}-2,m_{k}-2}(z))=-\mu_{k}$.
  \end{itemize}
It follows that $a_{c,k,m_{k}-2}=0$.

Iterating this procedure, we find $a_{c,k,m}=0$ for $m \in \{0,\ldots,m_{c,k}-1\}$. 

An obvious iteration of what precedes yields to $a_{c,j,m}  =0$ for all $j \in \{1,\ldots,k\}$ and all $m \in \{0,\ldots,m_{c,j}-1\}$, as expected. \\

In order to complete the proof, it remains to state and prove the following two lemmas used above. 

\begin{lem}\label{les logm sont indep}
The family $(\logm_{c,j})_{c \in \C^{\times}, j \geq 0}$ is $\Hahn$-linearly independent. 
\end{lem}

\begin{proof}
Assume on the contrary that the family $(\logm_{c,j})_{c \in \C^{\times}, j \geq 0}$ is $\Hahn$-linearly dependent.  
Consider a $\Hahn$-linearly dependent family $(\logm_{c,j})_{c \in C,j \geq 0}$ with $C \subset \C^{\times}$ finite (nonempty) of minimal cardinality.  There exist $c\in C$ and $j \geq 0$ such that $\logm_{c,j}$ is a $\Hahn$-linear combinaison of the $\logm_{d,k}$ with $d \in C$ such that $d \neq c$ or ($d=c$ and $k<j$). 

Let us first assume that $j=0$. So, we have 
\begin{equation}\label{eq ec lin comb}
 e_{c} =  \sum_{d\in C \setminus \{c\},k \geq 0}^{} \alpha_{d,k}(z)\logm_{d,k}
\end{equation}
for some $\alpha_{d,k}(z) \in \Hahn$. 
Applying $\malop{\ellmahl}$ to this equality, we obtain: 
\begin{equation}\label{eq ec lin comb phi}
c e_{c}
= \sum_{d\in C \setminus \{c\},k \geq 0}^{} \malop{\ellmahl}(\alpha_{d,k}(z))(d\logm_{d,k}+\logm_{d,k-1}).
\end{equation}
Considering the linear combinaison $\eqref{eq ec lin comb phi}-c\eqref{eq ec lin comb}$, we find 
\begin{multline}\label{lin com eq ec}
 0 = 
 \sum_{d\in C \setminus \{c\},k \geq 0}^{} (d \malop{\ellmahl}(\alpha_{d,k}(z))- c \alpha_{d,k}(z))\logm_{d,k}
 +\sum_{d\in C \setminus \{c\},k \geq 0}^{}\malop{\ellmahl}(\alpha_{d,k}(z))\logm_{d,k-1}\\
 =
  \sum_{d\in C \setminus \{c\},k \geq 0}^{} (d \malop{\ellmahl}(\alpha_{d,k}(z))- c \alpha_{d,k}(z)+\malop{\ellmahl}(\alpha_{d,k+1}(z)))\logm_{d,k}.
\end{multline}
We claim that this $\Hahn$-linear relation is nontrivial. Indeed, assume at the contrary that, for all $d \in C \setminus \{c\}$ and all $k \geq 0$, we have 
$$
d \malop{\ellmahl}(\alpha_{d,k}(z))- c \alpha_{d,k}(z)+\malop{\ellmahl}(\alpha_{d,k+1}(z))=0. 
$$
For $k$ large enough, we have $\alpha_{d,k+1}(z)=0$. But, if $\alpha_{d,k+1}(z)=0$ then it follows from Lemma \ref{lem:pas de car et log dans hahn} below that $\alpha_{d,k}(z)=0$. Iterating this, we find that all the $\alpha_{d,k}(z)$ with $d \in C \setminus \{c\}$ and $k \geq 0$ are zero, whence a contradiction.  So, the linear combinaison \eqref{lin com eq ec} is non trivial; this contradicts the minimality of $C$.

We now assume that $j \geq 1$. We have  
\begin{equation}\label{eq:lcj}
 \logm_{c,j} = \sum_{k=0}^{j-1} \alpha_{c,k}(z) \logm_{c,k}
+ \sum_{d\neq c,k \geq 0}^{} \alpha_{d,k}(z)\logm_{d,k}
\end{equation}
for some $\alpha_{d,k}(z) \in \Hahn$. 
Applying $\malop{\ellmahl}$ to this equality, we find: 
\begin{multline}\label{eq:phi lcj}
c\logm_{c,j}  +  \logm_{c,j-1}
=
 \sum_{k=0}^{j-1} \malop{\ellmahl}(\alpha_{c,k}(z) )(c\logm_{c,k}+\logm_{c,k-1})\\
 + \sum_{d\neq c,k \geq 0}^{} \malop{\ellmahl}(\alpha_{d,k}(z))(d\logm_{d,k}+\logm_{d,k-1}).
\end{multline}
Considering the linear combination \eqref{eq:phi lcj}$-c$\eqref{eq:lcj}, we get 
\begin{multline}
  \logm_{c,j-1} = 
  \sum_{k=0}^{j-1} c (\malop{\ellmahl}(\alpha_{c,k}(z) ) - \alpha_{c,k}(z)) \logm_{c,k}+
  \sum_{k=0}^{j-1} \malop{\ellmahl}(\alpha_{c,k}(z) )\logm_{c,k-1}\\
 + \sum_{d\neq c,k\geq 0}^{} (d\malop{\ellmahl}(\alpha_{d,k}(z))-c\alpha_{d,k}(z))\logm_{d,k}
 +\sum_{d\neq c,k\geq 0}^{} \malop{\ellmahl}(\alpha_{d,k}(z)) \logm_{d,k-1}.
\end{multline}
The above equality can be rewritten as 
\begin{multline}
(1-c(\malop{\ellmahl}(\alpha_{c,j-1}(z) ) - \alpha_{c,j-1}(z))) \logm_{c,j-1}\\
=
  \sum_{k=0}^{j-2} c (\malop{\ellmahl}(\alpha_{c,k}(z) ) - \alpha_{c,k}(z)) \logm_{c,k}+
  \sum_{k=0}^{j-1} \malop{\ellmahl}(\alpha_{c,k} (z))\logm_{c,k-1}\\
 + \sum_{d\neq c,k\geq 0}^{} (d\malop{\ellmahl}(\alpha_{d,k}(z))-c\alpha_{d,k}(z))\logm_{d,k}
 +\sum_{d\neq c,k\geq 0}^{} \malop{\ellmahl}(\alpha_{d,k}(z)) \logm_{d,k-1}.
\end{multline}
But $1-c(\malop{\ellmahl}(\alpha_{c,j-1}(z) ) - \alpha_{c,j-1}(z)) \neq 0$ (follows from Lemma \ref{lem:pas de car et log dans hahn} below), so we obtain that $\logm_{c,j-1}$ is a $\Hahn$-linear combinaison of the $\logm_{d,k}$ with $d \in C$ such that $d \neq c$ or ($d=c$ and $k<j-1$). Iterating this, we get that $e_{c,j}$ is a $\Hahn$-linear combinaison of the $\logm_{d,k}$ with $d \in C \setminus \{c\}$ and, hence, we are reduced to the first case considered at the beginning of this proof.
\end{proof}

\begin{lem}\label{lem:pas de car et log dans hahn}
Let $R$ be a ring. If  $g(z) \in \Hahn_{R}$ has a nonzero constant term, then the equation 
$$
\malop{\ellmahl}(f(z))-f(z)=g(z)
$$ 
has no solution $f(z) \in \Hahn_{R}$. 

Assume that $R$ is an integral domain. If $c,d$ are distinct nonzero elements of $R$, then the equation 
$$
c \malop{\ellmahl}(f(z))-df(z)=0
$$ 
has no nonzero solution $f(z) \in \Hahn_{R}$. 
\end{lem}

\begin{proof}
The first assertion follows from the fact that, for any $f(z)=\sum _{{\gamma\in \Q }}f_{\gamma}z^{\gamma} \in \Hahn_{R}$, the constant term of $\malop{\ellmahl}(f(z))-f(z)=\sum _{{\gamma\in \Q }}(f_{\gamma/\ellmahl}-f_{\gamma})z^{\gamma}$ is equal to $f_{0/\ellmahl}-f_{0}=0$.

Let us prove the second assertion. Consider $c,d$ as in the statement of the lemma and  let $f(z)=\sum _{{\gamma\in \Q }}f_{\gamma}z^{\gamma} \in \Hahn_{R}$ be such that $c \malop{\ellmahl}(f)-df=0$, {\it i.e.}, such that, for all $\gamma \in \Q$, $c f_{\gamma/\ellmahl}-df_{\gamma} =0$. If $f(z) \neq 0$, then there exists $\gamma \in \Q^{\times}$ such that $f_{\gamma} \neq 0$ and the latter equation implies that $f_{\ellmahl^{k} \gamma} \neq 0$ for all $k \in \Z$, {\it i.e.}, that $\ellmahl^{\Z}\gamma \subset \supp(f(z))$. This contradicts the fact that  $\supp(f(z))$ is well-ordered. 
\end{proof}

\bibliographystyle{alpha}
\bibliography{biblio}
\end{document}